\documentclass[pdflatex,sn-mathphys-num]{sn-jnl}

\usepackage{placeins}
\usepackage{graphicx}%
\usepackage{multirow}%
\usepackage{amsmath,amssymb,amsfonts}%
\usepackage{amsthm}%
\usepackage{mathrsfs}%
\usepackage[title]{appendix}%
\usepackage{xcolor}%
\usepackage{textcomp}%
\usepackage{manyfoot}%
\usepackage{booktabs}%
\usepackage{algorithm}%
\usepackage{algorithmicx}%
\usepackage{algpseudocode}%
\usepackage{listings}%
\usepackage{subcaption}
\usepackage{bm}
\usepackage{color}
\usepackage{epstopdf}
\usepackage{subcaption}
\usepackage{multirow}
\usepackage{dsfont}
\usepackage{booktabs}
\usepackage{longtable}

\theoremstyle{thmstyleone}%
\newtheorem{theorem}{Theorem}
\newtheorem{proposition}[theorem]{Proposition}%

\theoremstyle{thmstyletwo}%
\newtheorem{remark}{Remark}%

\theoremstyle{thmstylethree}%

\begin{document}

\title[Structure-Aware Variational State Preparation for Quantum Basket Option Pricing]{Structure-Aware Variational State Preparation for Quantum Basket Option Pricing}


\author[1]{\fnm{Dongwoo} \sur{Kim}}\email{dkim27@stevens.edu}

\author[1]{\fnm{Zhenyu} \sur{Cui}}\email{zcui6@stevens.edu}
\author[2,3,4]{\fnm{Daniel K.} \sur{Park}}\email{dkd.park@yonsei.ac.kr}
\author*[1]{\fnm{Chihoon} \sur{Lee}}\email{clee4@stevens.edu}

\affil*[1]{\orgdiv{School of Business}, \orgname{Stevens Institute of Technology}, \orgaddress{ \city{Hoboken}, \postcode{07030}, \state{New Jersey}, \country{USA}}}

\affil[2]{\orgdiv{Department of Statistics and Data Science}, \orgname{Yonsei University}, \orgaddress{\postcode{03722}, \state{Seoul}, \country{South Korea}}}

\affil[3]{\orgdiv{Department of Applied Statistics}, \orgname{Yonsei University}, \orgaddress{\postcode{03722}, \state{Seoul}, \country{South Korea}}}

\affil[4]{\orgdiv{Department of Quantum Information}, \orgname{Yonsei University}, \orgaddress{\postcode{03722}, \state{Seoul}, \country{South Korea}}}

\abstract{Basket option pricing often relies on Monte Carlo estimation, for which quantum
    amplitude estimation (QAE) provides a quadratic speed-up. However, the practical
    benefit of QAE can be limited by the depth of the state-preparation circuit. We
    propose a structure-aware quantum state-preparation framework for QAE-based
    basket option pricing. The framework uses tensor-train (TT) rank information to design shallow
    variational state-preparation circuits. In the independent regime, TT ranks remove
    unnecessary entangling links from a hardware-efficient ansatz. In correlated
    basket settings, we instead prepare asset-wise marginals locally and train a
    compact latent block to match the basket cumulative distribution function. The Basket-CDF objective targets the basket pushforward distribution rather than
    the full joint state, directly aligning state preparation with basket-dependent
    payoffs. Numerical experiments show that the proposed circuits replace the exponential
    state-preparation depth scaling of exact amplitude loading with linear scaling,
    while maintaining low-percent basket-pricing errors. Additional sampling-based training experiments and an end-to-end QAE integration study support compatibility with sample-estimated training and standard QAE-based pricing workflows.}




\maketitle

\section{Introduction}\label{sec1}

    Pricing multidimensional financial derivatives often requires repeated
    estimation of expectations under probability distributions. Classical Monte
    Carlo methods remain a standard tool for this task, especially in multi-asset
    settings where deterministic grid-based methods suffer from the curse of dimensionality.
    Quantum amplitude estimation (QAE) has therefore attracted  interest
    in quantum finance, since it provides a quadratic reduction in sampling
    complexity relative to classical Monte Carlo under suitable oracle assumptions
    \cite{brassard2002quantum}. This potential has motivated QAE-based algorithms
    for derivative pricing, including European options, basket options, and more
    general quantum-accelerated Monte Carlo pipelines
    \cite{woerner2019quantum,stamatopoulos2020option,herman2026quantum,yu2026quantum,blank2021quantum}.
    
    In financial applications of QAE, the discretized asset-price distribution
    must first be loaded into an amplitude-encoded probability state, which we
    refer to as the uncertainty state, before payoff encoding and amplitude estimation can be applied~\cite{woerner2019quantum,stamatopoulos2020option}.

    In multi-asset settings, this loading step becomes costly because the input
    state must represent a high-dimensional joint distribution and, in correlated
    baskets, its cross-asset dependence structure. When the assets
    are correlated, the state-preparation circuit must also encode cross-asset
    dependence. Exact
    amplitude loading can therefore lead to state-preparation circuits whose depth
    dominates the end-to-end pricing workflow.
    
    This bottleneck is particularly important because the uncertainty state must be
    prepared coherently whenever the pricing oracle is invoked. Thus, even if the
    amplitude-estimation stage is replaced by lower-depth variants
    \cite{suzuki2020amplitude,grinko2021iterative}, the overall
    workflow may still be limited by the depth of the state-preparation block. More
    broadly, input-preparation costs must be treated as first-order resource costs
    rather than secondary implementation details \cite{aaronson2015read}. On
    near-term or realistically noisy hardware, deep circuits further accumulate
    two-qubit gate errors, routing overhead, idle time, and decoherence
    \cite{preskill2018quantum,murali2019noise}. Efficient preparation of financial
    input distributions is therefore a central requirement for practical quantum
    Monte Carlo methods in derivative pricing and risk estimation
    \cite{carrera2021efficient,gomez2022survey,rattew2021efficient}.
    
    Several approaches have been proposed to make quantum state preparation more
    practical. Generative-model-based methods, including quantum generative adversarial
    networks and quantum circuit Born machines, learn approximate state-preparation
    circuits from probability data or samples
    \cite{benedetti2019generative,zoufal2019quantum}, and have been explored in
    quantum finance applications \cite{alcazar2022quantum}. These methods are flexible and
    hardware-compatible, but their performance depends strongly on the chosen
    ansatz, optimizer, and training objective. In particular, adversarial generative
    training can suffer from instability and mode-collapse behavior, while recent
    analyses of quantum generative models identify barren plateaus and loss
    concentration as potential trainability barriers
    \cite{tian2025quantum,rudolph2024trainability}. Tensor-network methods provide a more structured alternative. Matrix-product-state (MPS)
    and tensor-train (TT) representations can exploit low-rank or low-entanglement
    structure in the target amplitude tensor
    \cite{oseledets2011tensor,ran2020encoding}. MPS-based state
    preparation has been studied both generally and in QAE-based Monte Carlo contexts
    \cite{iaconis2024quantum,pereira2024encoding}. Related work also suggests that tensor-network representations can be combined
    with variational state-preparation strategies \cite{melnikov2023quantum}. This motivates using tensor structure as a structural guide for designing
    shallow variational circuits. In independent regimes, this idea is especially
    effective. TT ranks identify factorized cuts, so redundant entangling links can
    be removed without sacrificing the local expressivity needed to prepare each
    asset distribution. However, in correlated multi-asset settings, this sparsification benefit can
    largely disappear.

For correlated basket settings, accurate full joint-state preparation can be
unnecessarily demanding. When correlations make the TT ranks nontrivial across
most inter-asset cuts, the TT-informed entangling rule retains nearly all
nearest-neighbor links and effectively reduces to a linearly entangled
hardware-efficient ansatz. Although this remains much shallower than exact
amplitude loading, a shallow linear ansatz may not provide enough expressivity
to represent the full cross-asset dependence structure. Basket option prices,
however, depend on the asset vector only through the distribution of the basket
projection \cite{caldana2016general}. This suggests a separated construction:
use TT-informed circuits locally to prepare asset-wise marginals, and use an
additional latent register to learn cross-asset dependence only through its
effect on the basket distribution.

This paper develops such a structure-aware state-preparation framework for
reducing the circuit depth of distribution loading in QAE-based basket option pricing. The proposed marginal--latent Basket-CDF loader first prepares asset-wise marginals using TT-informed variational circuits, and then trains a compact latent dependence block to match the cumulative distribution function of the basket projection. This Basket-CDF objective is aligned with the pricing task, avoids full-state-fidelity supervision, and yields a circuit that can be reused across strikes and other basket-dependent functionals. More broadly, the same basket-projection viewpoint may be relevant for
portfolio-level risk functionals, such as tail probabilities, value-at-risk (VaR), and conditional value-at-risk (CVaR), whenever the target quantity is determined by the distribution of an aggregate portfolio value or loss. It also complements recent quantum option-pricing pipelines that incorporate
payoff information directly into the encoding and amplitude-estimation
procedure \cite{manzano2025alternative}.

The main contributions of this paper are threefold. First, we introduce a
marginal-TT latent Basket-CDF state-preparation circuit for correlated basket
settings. The construction prepares target asset-wise marginals locally, freezes
these marginal loaders, and trains a compact latent dependence block using a
basket-CDF objective with marginal regularization. This separates marginal
modeling from basket-relevant dependence learning and avoids full-state-fidelity
supervision. Second, we introduce the TT-informed variational state-preparation primitive
used for the local marginal loaders. The primitive uses TT-rank information as a
structural prior for entangling-topology design, rather than directly compiling
a full tensor-network representation. In independent multi-asset settings, the
same rule removes unnecessary cross-asset entanglers and yields a shallow
full-register validation case. Third, we analyze the state-preparation depth of the proposed constructions and
compare them with exact amplitude loading and direct TT/MPS baselines. Together
with basket-pushforward sufficiency and a CDF-based basket-call error bound,
these results show how tensor structure and basket-projection information can
be combined to construct shallow state-preparation circuits while preserving
downstream pricing accuracy.
    
    The remainder of this paper is organized as follows. Section~\ref{sec:background} reviews
    multi-asset uncertainty-state preparation, basket-structured financial
    functionals, hardware-efficient ansatzes, and tensor-train ranks. Section~\ref{sec:method}
    introduces the TT-informed state-preparation ansatz and the marginal-TT latent
    Basket-CDF construction. Section~\ref{sec:complexity} analyzes circuit-depth scaling and resource
    trade-offs. Section~\ref{sec:numerical_results} presents numerical experiments for independent and
    correlated multi-asset basket settings. Section~\ref{sec:discuss} concludes with limitations and
    future directions.

    \section{Background and Problem Formulation}
    \label{sec:background}
    
    \subsection{Multi-asset uncertainty-state preparation }
    \label{subsec:multi_asset_loading}
    
    We consider a discretized multi-asset financial model with $d$ underlying assets.
    Each asset is represented by $q$ uncertainty qubits, so that the local grid size is
    $m=2^q$ and the full asset register contains
    $n = dq$ qubits. A grid configuration is denoted by
    \[
        \bm{x}=(x_1,\ldots,x_d),
        \qquad x_i \in \{0,\ldots,m-1\},
    \]
    and the corresponding discretized joint probability mass function is denoted by
    $p(\bm{x})$. For pricing experiments, $p$ is interpreted as the discretized terminal distribution under the risk-neutral pricing measure.
    
    In amplitude-encoded quantum algorithms for finance, this distribution is represented
    by the uncertainty state
    \begin{equation}
        |\psi_p\rangle
        =
        \sum_{\bm{x}}
        \sqrt{p(\bm{x})}\,|\bm{x}\rangle .
        \label{eq:amplitude_encoding}
    \end{equation}
    Such uncertainty-state preparation is a standard component of quantum algorithms for
    risk analysis and derivative pricing based on amplitude estimation
    \citep{woerner2019quantum,stamatopoulos2020option}. Quantum approaches to derivative
    pricing have also been studied beyond simple one-dimensional Black--Scholes settings,
    including local-volatility models, Bermudan options, and alternative quantum-accelerated
    Monte Carlo pipelines
    \citep{kaneko2022quantum,miyamoto2022bermudan,manzano2025alternative}.
    
    This paper focuses on the uncertainty state preparation part of this workflow. In a
    generic multi-asset setting, the joint probability table has size
    \[
        m^d = 2^{dq},
    \]
    and exact amplitude loading, used here as exact state preparation of the full
    amplitude vector, can therefore become a dominant cost. This input-state preparation
    issue is a well-known obstacle in quantum algorithms that rely on classical data encoded
    into quantum states \citep{aaronson2015read,zoufal2019quantum,carrera2021efficient}. It is particularly
    important in quantum finance, where the same uncertainty state may be reused across
    many payoff evaluations but must still be prepared coherently each time the pricing
    oracle is invoked.
    
    These considerations motivate the use of structural information in the design of
    state-preparation circuits. Two types of structures will be important below. The first is
    tensor structure in the amplitude tensor, summarized by tensor-train ranks. The second
    is financial structure induced by basket projections, which can reduce the relevant
    distributional target from the full joint law to a lower-dimensional pushforward
    distribution.
    
    \subsection{Basket-structured financial functionals}
    \label{subsec:basket_functionals}
    
    We focus on basket-dependent financial functionals, where the payoff is obtained
    by first aggregating the multi-asset terminal state through a basket map and
    then applying a scalar payoff or risk function. Let
    $S_i(x_i)$ denote the terminal value of asset $i$ at local grid point $x_i$, and let
    $w_i$ be basket weights satisfying $\sum_{i=1}^d w_i=1$. We define the basket map
    \begin{equation}
        B(\bm{x})
        =
        \sum_{i=1}^{d} w_i S_i(x_i).
        \label{eq:basket_projection}
    \end{equation}
    Such basket-dependent functionals can then be written as
    \begin{equation}
        V_h
        =
        e^{-rT}\,
        \mathbb{E}_{\bm{x}\sim p}
        \left[
            h(B(\bm{x}))
        \right],
        \label{eq:basket_functional}
    \end{equation}
    where $h$ is a scalar payoff, risk, or indicator function. Here \(r\) is the continuously compounded risk-free rate and \(T\) is the
    maturity. The factor \(e^{-rT}\) discounts the terminal payoff expectation back
    to time zero.
    
    A basket option is a derivative whose payoff depends on a weighted portfolio,
or basket, of several underlying assets rather than on a single asset
\cite{caldana2016general}. For a basket call with strike $K$,
\[
    h_K(b) = (b-K)^+,
\]
so that
\begin{equation}
    V_K
    =
    e^{-rT}\,
    \mathbb{E}_{\bm{x}\sim p}
    \left[
        (B(\bm{x})-K)^+
    \right].
    \label{eq:basket_call}
\end{equation}
The same basket-functional formulation also covers basket puts, threshold
probabilities, index-linked payoffs, and risk measures based on the distribution
of a portfolio value. Thus, basket option pricing is used in this paper as a
benchmark application, but the underlying state-preparation problem is more
general. The objective is to prepare a quantum state that preserves the
distributional information relevant to basket-dependent functionals.
    
    In a standard quantum-pricing pipeline, the expectation in \eqref{eq:basket_functional}
    is estimated by combining uncertainty-state preparation, payoff encoding, and amplitude
    estimation \citep{brassard2002quantum, woerner2019quantum, stamatopoulos2020option}.
    Several lower-depth or modified amplitude-estimation variants have been proposed to
    reduce the resource requirements of the estimation stage
    \citep{suzuki2020amplitude,grinko2021iterative,manzano2023real,miyamoto2024bias}.
    Here, however, we treat amplitude estimation as a downstream estimation subroutine and
    focus on the state-preparation problem.
    
    The basket structure suggests that full joint-state fidelity is not always the most
    task-aligned target. A loader that approximates the full joint distribution may spend
    circuit resources on directions of the joint law that do not affect
    \eqref{eq:basket_functional}. Conversely, a loader that accurately reproduces the
    distribution of the scalar random variable $B(\bm{x})$ can be accurate for a broad
    family of downstream basket-structured quantities even if it does not recover the full
    joint distribution pointwise.
    
    For later use, we denote by
    \[
        B_{\#}p
    \]
    the pushforward distribution of $p$ under the basket map $B$. Equivalently, for any
    Borel set $A\subset \mathbb{R}$,
    \[
        (B_{\#}p)(A)
        =
        p\!\left(\{\bm{x}: B(\bm{x})\in A\}\right).
    \]
    This notation separates the full joint distribution $p$ from the one-dimensional
    distribution of the basket value. Section~\ref{sec:method} uses this distinction to
    formulate a projection-based loading objective.
    
    \subsection{Hardware-efficient ansatzes for quantum state preparation}
    \label{subsec:hea}
    
    Hardware-efficient ansatzes are variational circuits constructed from repeated layers
    of single-qubit rotations and entangling gates, typically chosen to be compatible with
    hardware-native gate sets and limited connectivity
    \citep{kandala2017hardware,cerezo2021variational}. In state-preparation problems, an
    ansatz may be viewed as being specified by three main design choices: the local rotation
    block, the entangling topology, and the number of layers.
    
    These choices play different roles. The rotation block controls local amplitude
    adjustments, while the entangling topology determines how correlations are distributed
    across the qubit register. Increasing the number of layers or using a denser entangling
    graph enlarges the variational family, but it also increases two-qubit gate count and
    transpiled circuit depth. Since two-qubit gates, routing overhead, and idle-time
    decoherence are among the main practical costs in near-term implementations
    \citep{preskill2018quantum,murali2019noise}, the entangling topology
    should be treated as a structural modeling choice rather than a minor implementation
    detail.
    
    The distribution-loading setting considered here has an additional simplification:
    the target amplitudes in \eqref{eq:amplitude_encoding} are real and nonnegative.
    Accordingly, the central challenge is not to represent arbitrary complex phases, but
    to approximate the amplitude profile and its induced correlation structure. This makes
    real-amplitude hardware-efficient circuits based on $R_y$ rotations and CNOT entanglers
    a natural variational template. The key question is then where entangling resources
    should be placed.
    
    A generic entangling pattern may allocate two-qubit gates inefficiently when the target
    distribution has independent, localized, or otherwise nonuniform dependence structure.
    The TT-informed construction in Section~\ref{subsec:tt_informed_hea} addresses this issue
    by using tensor-train rank information to determine which nearest-neighbor entangling
    links are structurally nontrivial.
    
    \subsection{Tensor-train ranks as structural information}
    \label{subsec:tt}
    
    Tensor-train (TT) decomposition provides a compact representation of a high-dimensional
    tensor by factorizing it into a sequence of low-order cores
    \citep{oseledets2011tensor}. For a tensor $A(i_1,\ldots,i_D)$, the TT representation is
    \begin{equation}
        A(i_1,\ldots,i_D)
        =
        G^{(1)}(i_1)
        G^{(2)}(i_2)
        \cdots
        G^{(D)}(i_D),
        \label{eq:tt_decomposition}
    \end{equation}
    where each $G^{(k)}(i_k)$ is a matrix-valued core, except at the boundaries. The
    intermediate dimensions
    \begin{equation}
        r_1,\ldots,r_{D-1}
        \label{eq:tt_rank_profile}
    \end{equation}
    are the TT ranks.
    
    The TT-rank profile carries structural information about the tensor. Each rank $r_k$
    measures the complexity across the consecutive bipartition
    \begin{equation}
        (i_1,\ldots,i_k)\,|\,(i_{k+1},\ldots,i_D).
        \label{eq:tt_consecutive_cut}
    \end{equation}
    A trivial rank indicates that the tensor factorizes across the corresponding cut,
    whereas a larger rank indicates nontrivial dependence between the two sides of the
    tensor chain.
    
    For quantum state preparation, the amplitude vector
    $\{\sqrt{p(\bm{x})}\}_{\bm{x}}$ can be reshaped into a qubit-index tensor,
    \begin{equation}
        A(i_1,\ldots,i_n)
        =
        \sqrt{p(i_1,\ldots,i_n)},
        \qquad
        i_j\in\{0,1\},
        \label{eq:qubit_index_amplitude_tensor}
    \end{equation}
    where the multi-index $(i_1,\ldots,i_n)$ represents the computational-basis state
    under a chosen qubit ordering. The resulting TT ranks summarize how amplitude
    correlations are distributed along this ordering. Tensor-network and
    matrix-product-state ideas have been used in several quantum state-preparation
    settings
    \citep{ran2020encoding,melnikov2023quantum,iaconis2024quantum,pereira2024encoding}.
    
    Since the TT-rank profile depends on the ordering of tensor modes, qubit ordering
    becomes part of the circuit-design problem. This is consistent with the broader
    observation that correlation-informed qubit permutations can reduce ansatz depth in
    variational quantum circuits \citep{tkachenko2021correlation}. The methods below use
    TT information as a structural prior rather than as a direct state-synthesis prescription.

    \section{Methodology}
    \label{sec:method}
    
    This section presents the proposed Basket-CDF state-preparation framework. We first describe the TT-informed marginal primitive and then introduce the marginal-TT + latent Basket-CDF loader for correlated basket settings.

    \subsection{TT-informed state-preparation primitive}
    \label{subsec:tt_informed_hea}
    
    We first describe the TT-informed state-preparation primitive used for the
    asset-wise marginal circuits in the proposed Basket-CDF framework. The same
    construction can also be applied to the full asset register in the independent
    regime, where the joint distribution factorizes across asset blocks. In
    correlated basket settings, however, we use this primitive locally for marginal
    preparation and learn cross-asset dependence through the latent Basket-CDF
    construction introduced in the next subsection.
    
    Consider a target probability distribution $p$ on a $N$-qubit register, with
    amplitude-encoded state
    \[
        |\psi_p\rangle
        =
        \sum_x \sqrt{p(x)}\,|x\rangle .
    \]
    For the marginal circuits used later, $N=q$ and $p=p_a$ is the one-asset
    marginal distribution of asset $a$. We reshape the amplitude vector
    $\{\sqrt{p(x)}\}_x$ into a qubit-index tensor of shape $2\times\cdots\times 2$
    and apply TT-SVD to obtain the rank profile $(r_1,\ldots,r_{N-1})$. We use this
    TT-rank profile only as a structural prior for the variational ansatz, rather
    than directly compiling the tensor-network representation into a circuit.
    
    The TT-rank profile defines the nearest-neighbor entangling set
    \begin{equation}
        E_{\mathrm{TT}}
        =
        \{(k-1,k): r_k>1,\; k=1,\ldots,N-1\}.
        \label{eq:ett}
    \end{equation}
    Thus, an entangling link is retained only when the corresponding TT rank is
    nontrivial. Based on this graph, we define a $L$-layer real-amplitude
    hardware-efficient ansatz
    \begin{equation}
        U_{\mathrm{TT}}(\theta)
        =
        \prod_{\ell=1}^{L}
        \left(
            U_{\mathrm{ent}}^{(\ell)}(E_{\mathrm{TT}})
            U_{\mathrm{rot}}^{(\ell)}(\theta)
        \right),
        \label{eq:tt_informed_ansatz}
    \end{equation}
    where
    \begin{equation}
        U_{\mathrm{rot}}^{(\ell)}(\theta)
        =
        \prod_{j=1}^{N} R_y(\theta_{j,\ell}),
        \qquad
        U_{\mathrm{ent}}^{(\ell)}(E_{\mathrm{TT}})
        =
        \prod_{(i,j)\in E_{\mathrm{TT}}}
        \mathrm{CNOT}_{i\rightarrow j}.
        \label{eq:tt_layers}
    \end{equation}
    The use of $R_y$ rotations is matched to the real nonnegative amplitude
    structure of probability loading. The main structural design variable is
    therefore the entangling graph.
    
    The primitive is trained by minimizing the fidelity loss
    \begin{equation}
        \mathcal{L}_{\mathrm{fid}}(\theta)
        =
        1-
        \left|
            \langle \psi_p|
            U_{\mathrm{TT}}(\theta)|0\rangle^{\otimes N}
        \right|^2 .
        \label{eq:full_state_fidelity_loss}
    \end{equation}
    In the proposed Basket-CDF loader, this fidelity objective is used only for the
    low-dimensional asset-wise marginal circuits. The latent dependence block is
    trained by the Basket-CDF objective introduced below and does not use full-state
    fidelity.
    
    By the TT-rank interpretation in Section~\ref{subsec:tt}, a rank-one
    consecutive cut indicates factorization across that cut. The following result
    connects this observation to the TT-informed entangling rule
    \eqref{eq:ett}.
    
    \begin{proposition}[Product structure and TT-informed entanglers]
    \label{prop:product_tt_rank}
    Suppose that the target amplitude tensor factorizes across a consecutive cut,
    \[
        A(i_1,\ldots,i_N)
        =
        A_L(i_1,\ldots,i_k)
        A_R(i_{k+1},\ldots,i_N).
    \]
    Then the TT rank across the cut
    \[
        (i_1,\ldots,i_k)\,|\,(i_{k+1},\ldots,i_N)
    \]
    is equal to one. Consequently, the TT-informed entangling rule
    \eqref{eq:ett} removes the nearest-neighbor entangling edge associated with
    this cut.
    \end{proposition}
    
    \begin{proof}
    The unfolding of $A$ across the stated cut is the outer product of the
    vectorized tensors $A_L$ and $A_R$, and therefore has matrix rank one. Since the
    TT rank across a consecutive cut equals the rank of the corresponding unfolding,
    the TT rank is one. The associated edge is therefore excluded by
    \eqref{eq:ett}.
    \end{proof}
    
    In the independent multi-asset regime,
    \[
        p(\bm{x})=\prod_{a=1}^{d}p_a(x_a),
        \qquad
        \sqrt{p(\bm{x})}
        =
        \prod_{a=1}^{d}\sqrt{p_a(x_a)}.
    \]
    Thus, inter-asset cuts have TT rank one and do not require cross-asset
    entanglers. This independent full-register case is used only as a validation
of the TT-informed primitive, and the      corresponding numerical results are
reported in Appendix~\ref{app:tt_independent}. In the main correlated-basket
construction, the same TT-informed rule is applied locally to the asset-wise
marginal loaders introduced below.
    
  \begin{figure}[t]
    \centering \includegraphics[width=1.0\linewidth] {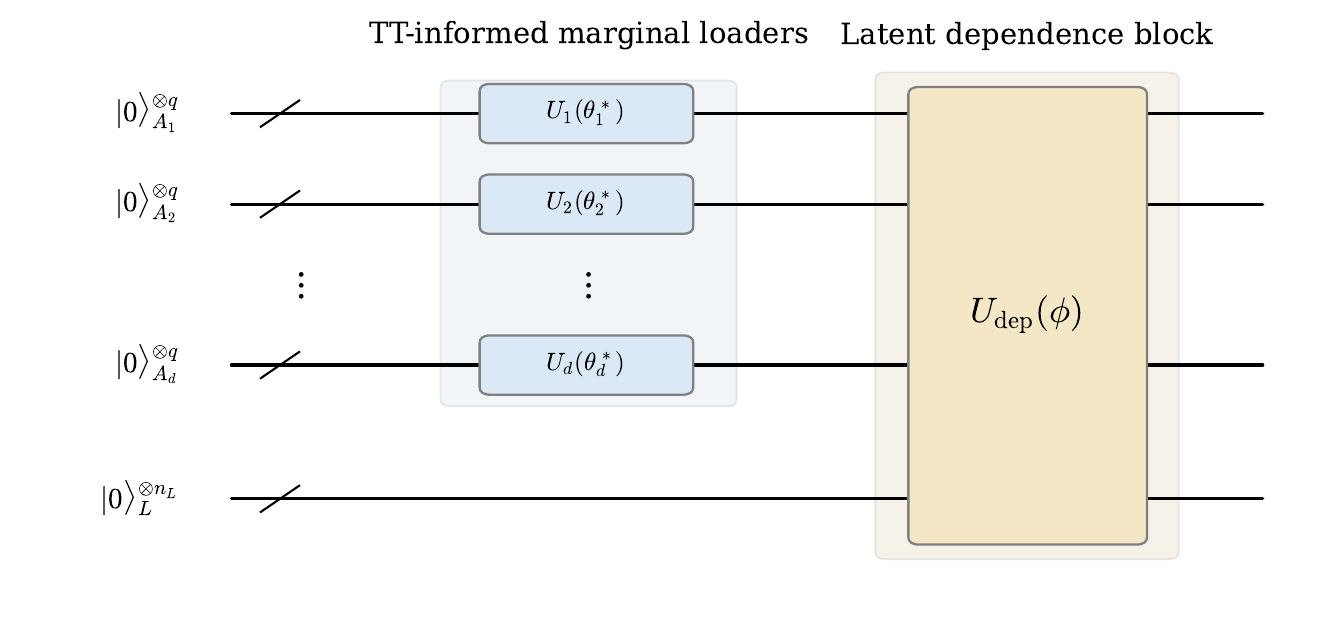}
    \caption{General circuit architecture of the marginal-TT loader with a latent dependence block. The registers $A_1,\ldots,A_d$ denote the $d$ asset registers, with $q$ uncertainty qubits assigned to each asset, and $L$ denotes an auxiliary latent register with $n_L$ qubits. $U_a(\theta_a^\ast)$ is a TT-informed marginal loader trained to prepare the marginal distribution of asset $a$. The latent dependence block $U_{\mathrm{dep}}(\phi)$ is then applied to the prepared asset registers and the latent register. If the final state is $|\Psi_\phi\rangle$, then the asset-register measurement distribution is obtained by measuring in the computational basis and marginalizing over the latent outcome: $\widetilde p_\phi(\bm{x})=\sum_{\ell} |\langle \bm{x},\ell|\Psi_\phi\rangle|^2$.}
\label{fig:general_loader}
\end{figure}
    
    \subsection{Marginal-TT Loader with a Latent Dependence Block}
    \label{subsec:marginal_latent_loader}
    
    For each asset $a=1,\ldots,d$, define the asset-wise marginal distribution
    \begin{equation}
        p_a(x_a)
        =
        \sum_{\bm{x}_{-a}} p(\bm{x}),
        \label{eq:asset_marginal}
    \end{equation}
    where $\bm{x}_{-a}$ denotes all asset coordinates except $x_a$. When the full
    target probability table is available, this marginal is obtained by reshaping
    the probability vector into a $d$-way tensor of shape
    \[
        (2^q,\ldots,2^q)
    \]
    and summing out all tensor axes except the axis corresponding to asset $a$.
    Equivalently,
    \begin{equation}
        p_a(x_a)
        =
        \sum_{x_1=0}^{2^q-1}
        \cdots
        \sum_{x_{a-1}=0}^{2^q-1}
        \sum_{x_{a+1}=0}^{2^q-1}
        \cdots
        \sum_{x_d=0}^{2^q-1}
        p(x_1,\ldots,x_d).
        \label{eq:asset_marginal_explicit}
    \end{equation}
    Thus, the marginal stage keeps the calibrated single-asset distributions of the
    target joint law while discarding cross-asset dependence at initialization.
    
    The corresponding product-marginal distribution is
    \begin{equation}
        p_{\mathrm{prod}}(\bm{x})
        =
        \prod_{a=1}^{d} p_a(x_a).
        \label{eq:product_marginal_distribution}
    \end{equation}
    By construction, $p_{\mathrm{prod}}$ has the same asset-wise marginals as
    $p$, but it does not contain the cross-asset dependence structure of the target
    joint distribution.
    
    The TT-informed ansatz of Section~\ref{subsec:tt_informed_hea} is applied
    locally to each marginal $p_a$. Let $U_a(\theta_a)$ denote the trained
    $q$-qubit marginal loader satisfying
    \begin{equation}
        U_a(\theta_a)|0\rangle^{\otimes q}
        \approx
        \sum_{x_a}
        \sqrt{p_a(x_a)}\,|x_a\rangle .
        \label{eq:asset_marginal_loader}
    \end{equation}
    The product-marginal initialization circuit is then
    \begin{equation}
        U_{\mathrm{marg}}
        =
        \bigotimes_{a=1}^{d} U_a(\theta_a).
        \label{eq:product_marginal_circuit}
    \end{equation}
    Although $p_{\mathrm{prod}}$ is a distribution over the full $2^{dq}$-point
    grid, it is represented through $d$ marginal distributions of size $2^q$.
    Therefore, the classical distributional input for this stage scales as
    \[
        O(d2^q),
    \]
    rather than $O(2^{dq})$.
    
    This construction separates marginal modeling from dependence modeling in an
    operational sense. The circuits $U_a(\theta_a)$ are trained first to prepare the
    asset-wise marginals in \eqref{eq:asset_marginal} and are then frozen. Their
    tensor product \eqref{eq:product_marginal_circuit} initializes the asset
    register in an approximation of the product-marginal distribution
    \eqref{eq:product_marginal_distribution}. In the second stage, the marginal
    loaders are not retrained; only the parameters of the latent dependence block
    are optimized.
    
    In correlated regimes, $U_{\mathrm{marg}}$ is therefore not intended to
    represent the full joint law. It represents the product of the calibrated
    asset-wise marginals. Cross-asset dependence is introduced only after this
    initialization, through a trainable asset--latent interaction. This
    marginal--dependence separation is conceptually related to copula-based
    financial modeling and quantum copula-based risk aggregation, where marginal
    distributions and dependence structure are treated as distinct modeling
    components \citep{mori2024quantum, zhu2023copula, zhu2022generative}. We do not explicitly construct a copula in this work. Instead, the
    latent register is used as a circuit-level dependence-correction channel: the
    marginal circuits encode the asset-wise distributions, while the subsequent
    asset--latent interaction modifies the joint asset-register distribution in a
    way that is trained through the basket pushforward objective.
    
    Let $A$ denote the $dq$-qubit asset register and let $L$ denote an auxiliary
    latent register with $n_L$ qubits. Starting from the product-marginal state, the
    dependence block prepares
    \begin{equation}
        |\Psi_{\phi}\rangle
        =
        U_{\mathrm{dep}}(\phi)
        \left(
            U_{\mathrm{marg}}|0\rangle_A
            \otimes
            |0\rangle_L
        \right).
        \label{eq:latent_state}
    \end{equation}
    The induced distribution on the asset register is obtained by marginalizing over the latent register:
\begin{equation}
    \widetilde p_{\phi}(\bm{x})
    =
    \sum_{\ell\in\{0,1\}^{n_L}}
    \left|
        \langle \bm{x},\ell|\Psi_{\phi}\rangle
    \right|^2 .
    \label{eq:latent_marginal_distribution}
\end{equation}
Thus, the latent block converts the product-marginal initialization into
a new asset-register distribution $\widetilde p_{\phi}$. The latent register is
used only as a circuit-level auxiliary channel for dependence correction; it is
not part of the financial payoff register and should not be interpreted as a
financial factor model.

Figure~\ref{fig:general_loader} summarizes the general circuit architecture
for an arbitrary number of assets. Each asset register is first prepared by a
TT-informed marginal loader, and the resulting product-marginal state is
passed, together with an auxiliary latent register, through the latent
dependence block. The latent outcome is marginalized out when computing the
asset-register measurement distribution. Figure~\ref{fig:marginal_latent_loader}
then shows a two-asset gate-level implementation example with $q=3$
uncertainty qubits per asset.

The Basket-CDF training objective introduced in the next subsection learns
this dependence correction through its effect on the induced basket
pushforward distribution, while a marginal penalty discourages distortion of
the prepared asset-wise marginals.
    
    \begin{figure}[t]
        \centering
        \includegraphics[width=\linewidth]{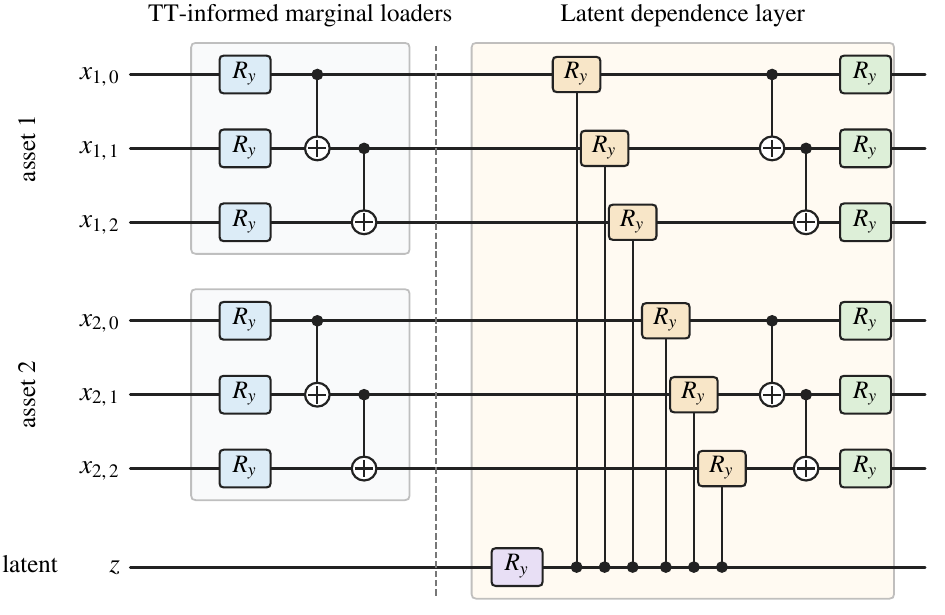}
        \caption{
        Marginal-TT + latent basket-projection loader for a two-asset example with $q=3$
        uncertainty qubits per asset. TT-informed marginal loaders act locally on each
        asset register, followed by a latent-conditioned dependence layer. The auxiliary
        latent qubit $z$ controls rotations on the asset qubits, while CNOT chains are
        restricted to individual asset registers.
    } \label{fig:marginal_latent_loader}
    \end{figure}
    
    In the numerical implementation, we use a single latent qubit unless stated
    otherwise. As illustrated in Figure~\ref{fig:marginal_latent_loader}, a
    dependence layer applies an $R_y$ rotation on the latent qubit,
    latent-controlled $R_y$ rotations from the latent qubit to the asset qubits,
    nearest-neighbor CNOT chains within each asset register, and local $R_y$
    rotations on the asset qubits. This defines a compact dependence-correction
    block whose number of variational angles and controlled rotations scales
    linearly in $dq$ for fixed latent-register size.

    \subsection{Basket-CDF objective and projection justification}
    \label{subsec:basket_cdf_objective}
    
    Let
    \begin{equation}
        B(\bm{x})
        =
        \sum_{a=1}^{d} w_a S_a(x_a),
        \qquad
        \sum_{a=1}^{d} w_a=1,
        \label{eq:method_basket_projection}
    \end{equation}
    be the basket projection. Using the pushforward notation introduced in Section~\ref{subsec:basket_functionals}, we denote by $B_{\#}p$ and $B_{\#}\widetilde p_{\phi}$ the basket-value distributions induced by the target asset-register distribution $p$ and the model asset-register distribution $\widetilde p_{\phi}$, respectively. Equivalently, these are the pushforward distributions of the corresponding asset-register laws under the basket map $B$.
    In the basket-option literature, the same object is commonly described as the
    distribution or law of the basket value, since the payoff is written on a weighted
    sum of the underlying assets \citep{caldana2016general}.
    
    The corresponding cumulative distribution functions are
    \begin{equation}
        F_p(b)
        =
        \Pr_{\bm{x}\sim p}[B(\bm{x})\le b],
        \qquad
        F_{\phi}(b)
        =
        \Pr_{\bm{x}\sim \widetilde p_{\phi}}[B(\bm{x})\le b].
        \label{eq:target_model_cdf}
    \end{equation}
    The basket-CDF loss matches these one-dimensional distributions:
    \begin{equation}
        \mathcal{L}_{\mathrm{CDF}}(\phi)
        =
        \frac{1}{J}
        \sum_{j=1}^{J}
        \left(
            F_{\phi}(b_j)-F_p(b_j)
        \right)^2,
        \label{eq:cdf_loss}
    \end{equation}
    where $\{b_j\}_{j=1}^{J}$ are basket-value grid points or bin boundaries. In the
    grid-based experiments, these points are chosen from the ordered distinct basket
    values induced by the discretized asset grid. In sample-based implementations, they
    may instead be chosen as fixed grid points or empirical quantiles. The loss in
    \eqref{eq:cdf_loss} is a discretized CDF-matching criterion, closely related in
    spirit to CDF-based scoring rules for probability distributions, such as the
    continuous ranked probability score
    \citep{matheson1976scoring,gneiting2007strictly}.
    
    The dependence block is trained using
    \begin{equation}
        \mathcal{L}_{\mathrm{BP}}(\phi)
        =
        \lambda_B \mathcal{L}_{\mathrm{CDF}}(\phi)
        +
        \lambda_M \mathcal{L}_{\mathrm{marg}}(\phi).
        \label{eq:basket_projection_loss}
    \end{equation}
    Here $\mathcal{L}_{\mathrm{marg}}$ penalizes distortion of the asset-wise marginals
    after the dependence block. More explicitly, if
    $\widetilde p_{\phi,a}$ denotes the $a$th asset-wise marginal of
    $\widetilde p_{\phi}$, then we use
    \begin{equation}
        \mathcal{L}_{\mathrm{marg}}(\phi)
        =
        \frac{1}{d}
        \sum_{a=1}^{d}
        \frac{1}{2^q}
        \sum_{x_a=0}^{2^q-1}
        \left(
            \widetilde p_{\phi,a}(x_a)-p_a(x_a)
        \right)^2.
        \label{eq:marginal_loss}
    \end{equation}
    The defining term is $\mathcal{L}_{\mathrm{CDF}}$, which matches the
    basket-projection distribution directly.
    
    This task-aligned perspective is related to recent payoff-aware QAE-based Monte Carlo pipelines.
    For example, ~\citet{manzano2025alternative} incorporate
    payoff-relevant information directly into the pricing circuit through payoff
    encoding and modified Real Quantum Amplitude Estimation. Our approach applies a
    complementary form of task alignment at the state-preparation stage. Rather than
    encoding a particular payoff value or optimizing a single option price, we match
    the distribution of the basket projection itself. The resulting loader remains
    strike-independent and can be reused across different strikes and other
    functionals of the form $\mathbb{E}[h(B)]$.
    
    For diagnostics, we may additionally report the classical fidelity
    \begin{equation}
        F(p,\widetilde p_{\phi})
        =
        \left(
            \sum_{\bm{x}}
            \sqrt{p(\bm{x})\widetilde p_{\phi}(\bm{x})}
        \right)^2 .
        \label{eq:classical_fidelity_diagnostic}
    \end{equation}
    This fidelity is not the defining objective of the basket-projection loader. It is
    used only to assess how closely the learned distribution agrees with the full joint
    law.
    
    The basket-projection objective is justified by the fact that any financial quantity
    depending on the asset vector only through $B(\bm{x})$ is fully determined by the
    pushforward distribution $B_{\#}p$.
    
    \begin{proposition}[Basket-pushforward sufficiency]
    \label{prop:basket_pushforward_sufficiency}
    Let $p$ and $r$ be two probability distributions on the same discretized asset grid.
    If
    \[
        B_{\#}p = B_{\#}r,
    \]
    then for any measurable scalar function $h$,
    \begin{equation}
        \mathbb{E}_{\bm{x}\sim p}[h(B(\bm{x}))]
        =
        \mathbb{E}_{\bm{x}\sim r}[h(B(\bm{x}))].
        \label{eq:pushforward_sufficiency}
    \end{equation}
    In particular, all basket-call prices are equal for all strikes $K$.
    \end{proposition}
    
    \begin{proof}
    Since $h(B(\bm{x}))$ is the composition $h\circ B$, we have
    \[
        \mathbb{E}_{\bm{x}\sim p}[h(B(\bm{x}))]
        =
        \int h(b)\,d(B_{\#}p)(b).
    \]
    If $B_{\#}p=B_{\#}r$, then the same integral is obtained under $r$. Taking
    $h(b)=(b-K)^+$ gives the basket-call case.
    \end{proof}
    
    Proposition~\ref{prop:basket_pushforward_sufficiency} shows that full joint recovery
    is stronger than necessary for basket-structured financial quantities. The next
    result gives the corresponding error relation for basket calls.
    
    \begin{proposition}[CDF error and basket-call error]
    \label{prop:cdf_error_call_error}
    Assume that the basket value is supported on $[B_{\min},B_{\max}]$. Let $C_p(K)$ and
    $C_r(K)$ denote the undiscounted basket-call expectations under $p$ and $r$:
    \[
        C_p(K)=\mathbb{E}_{p}[(B-K)^+],
        \qquad
        C_r(K)=\mathbb{E}_{r}[(B-K)^+].
    \]
    If $K < B_{\max}$, then
    \begin{equation}
        |C_p(K)-C_r(K)|
        \le
        \int_K^{B_{\max}}
        |F_p(t)-F_r(t)|\,dt.
        \label{eq:cdf_call_bound}
    \end{equation}
    If $K\ge B_{\max}$, both call expectations are zero.
    \end{proposition}
    
    \begin{proof}
    Using the tail representation of call payoffs,
    \[
        \mathbb{E}[(B-K)^+]
        =
        \int_K^{B_{\max}} \Pr(B>t)\,dt
        =
        \int_K^{B_{\max}} (1-F_B(t))\,dt
    \]
    for $K<B_{\max}$. Taking the difference between the two expectations and applying
    the triangle inequality gives \eqref{eq:cdf_call_bound}. If $K\ge B_{\max}$, the
    payoff is identically zero under both distributions.
    \end{proof}
    
    \begin{remark}[Resolution-error decomposition]
    \label{rem:resolution_error_decomposition}
    The CDF-based error bound also clarifies the effect of increasing the number of
    uncertainty qubits per asset. Let $V(K)$ denote the continuous basket-call value,
    let $V_q(K)$ denote the exact value under the $q$-qubit discretized target
    distribution, and let $V_{\phi,q}(K)$ denote the value induced by the
    Basket-CDF loader on the same discretized grid. Then
    \[
        |V(K)-V_{\phi,q}(K)|
        \le
        |V(K)-V_q(K)|
        +
        |V_q(K)-V_{\phi,q}(K)|.
    \]
    The first term is the discretization error. Since the basket-call payoff
    $(B-K)^+$ is Lipschitz in the basket value, this term decreases as the asset grid
    is refined, up to the truncation and quadrature error of the chosen
    discretization scheme.
    
    The second term is the loader-induced error on the discretized problem. By
    Proposition~\ref{prop:cdf_error_call_error},
    \[
        |V_q(K)-V_{\phi,q}(K)|
        \le
        e^{-rT}
        \int_K^{B_{\max}}
        |F_q(t)-F_{\phi,q}(t)|\,dt .
    \]
    Thus, for a fixed discretized problem, basket-call accuracy is controlled by the
    basket-CDF discrepancy rather than by full joint-distribution fidelity.
    
    This decomposition suggests why increasing $q$ can improve the proposed
    Basket-CDF loader: a larger $q$ gives a finer representation of the marginal
    asset grids and of the induced basket distribution. The improvement is not
    expected to be strictly monotone for every strike, dimension, or correlation
    structure, because changing $q$ also changes the discretized target distribution,
    the variational landscape, and the circuit size.
    \end{remark}
    
    Finally, full-state fidelity is useful as a diagnostic, but it is not necessary for
    basket-projection accuracy.

    \begin{remark}[Basket accuracy without full-state fidelity]\label{rem:payoff_without_fidelity}
    Full-state fidelity is only a diagnostic for the Basket-CDF loader. If two
    configurations $\bm{x}\neq\bm{y}$ have the same basket value,
    $B(\bm{x})=B(\bm{y})$, then moving probability mass between these two
    configurations leaves the basket pushforward distribution unchanged. Therefore,
    all quantities of the form $\mathbb{E}[h(B)]$ remain unchanged, even though the
    full joint distribution and hence the classical fidelity may change. This
    illustrates why optimizing the basket-projection loss can be sufficient for
    basket-structured pricing tasks.
    \end{remark}
    
    Together, Remarks~\ref{rem:resolution_error_decomposition}
    and~\ref{rem:payoff_without_fidelity} show that the Basket-CDF loader should be
    evaluated through the basket pushforward distribution rather than full-state
    fidelity. The former links CDF discrepancy to pricing error, while the latter
    explains why different joint distributions can induce the same basket-structured
    quantities. The complete Basket-CDF training pipeline is summarized in
    Algorithm~\ref{alg:basket_cdf_state_preparation}. 
    
    \begin{algorithm}[H]
    \caption{Two-stage Basket-CDF state-preparation procedure}
    \label{alg:basket_cdf_state_preparation}
    \begin{algorithmic}[1]
    \Statex \textbf{Input:} Target PMF or target samples; basket map $B$; CDF grid
    $\{b_j\}_{j=1}^{J}$; weights $\lambda_B,\lambda_M$; latent size $n_L$;
    depth $L_{\mathrm{dep}}$; optimizer budget $T_{\mathrm{opt}}$
    \Statex \textbf{Output:} Trained state-preparation circuit $U_{\mathrm{CDF}}$
    
    \State Compute or estimate asset-wise marginals $\{p_a\}_{a=1}^{d}$
    \State Compute or estimate target basket CDF values $\{F_p(b_j)\}_{j=1}^{J}$
    
    \For{$a=1,\ldots,d$}
        \State Apply TT-SVD to $\sqrt{p_a}$ and build the marginal ansatz $U_a(\theta_a)$
        \State Train $U_a(\theta_a)$ by marginal fidelity loss to obtain $\theta_a^*$
    \EndFor
    
    \State Set $U_{\mathrm{marg}}=\bigotimes_{a=1}^{d} U_a(\theta_a^*)$ and freeze the marginal loaders
    \State Append an $n_L$-qubit latent register and initialize $U_{\mathrm{dep}}(\phi)$
    
    \State Define $L_{\mathrm{BP}}(\phi)$ by preparing $|\Psi_\phi\rangle$,
    tracing out the latent register to obtain $\tilde p_\phi$, computing
    $\{F_\phi(b_j)\}_{j=1}^{J}$, and evaluating
    \[
    L_{\mathrm{BP}}(\phi)
    =
    \lambda_B
    \frac{1}{J}\sum_{j=1}^{J}
    \left(F_{\phi}(b_j)-F_p(b_j)\right)^2
    +
    \lambda_M L_{\mathrm{marg}}(\phi),
    \]
    where $L_{\mathrm{marg}}$ penalizes distortion of the asset-wise marginals.
    
    \State Optimize only $\phi$ for at most $T_{\mathrm{opt}}$ iterations to obtain $\phi^*$
    
    \State \Return
    $U_{\mathrm{CDF}}
    =
    U_{\mathrm{dep}}(\phi^*)
    \left(U_{\mathrm{marg}}\otimes I_L\right)$
    \end{algorithmic}
    \end{algorithm}
    
    \FloatBarrier

    \subsection{Sample-estimable formulation}
    \label{subsec:sample_estimable}
    
    The basket-CDF objective can be estimated from samples. Given target-model samples
    \[
        \bm{x}^{(1)},\ldots,\bm{x}^{(N)}\sim p
    \]
    and measurement samples
    \[
        \bm{y}^{(1)},\ldots,\bm{y}^{(M)}\sim \widetilde p_{\phi}
    \]
    from the variational loader, define the empirical CDFs
    \begin{equation}
        \widehat F_p(b)
        =
        \frac{1}{N}
        \sum_{n=1}^{N}
        \mathbf{1}\{B(\bm{x}^{(n)})\le b\},
        \label{eq:empirical_target_cdf}
    \end{equation}
    and
    \begin{equation}
        \widehat F_{\phi}(b)
        =
        \frac{1}{M}
        \sum_{m=1}^{M}
        \mathbf{1}\{B(\bm{y}^{(m)})\le b\}.
        \label{eq:empirical_model_cdf}
    \end{equation}
    The sample-based CDF loss is
    \begin{equation}
        \widehat{\mathcal{L}}_{\mathrm{CDF}}(\phi)
        =
        \frac{1}{J}
        \sum_{j=1}^{J}
        \left(
            \widehat F_{\phi}(b_j)-\widehat F_p(b_j)
        \right)^2 .
        \label{eq:sample_cdf_loss}
    \end{equation}
    
    The sample-based formulation is supported by standard empirical-CDF
    concentration. By the Dvoretzky--Kiefer--Wolfowitz inequality with Massart's
    sharp constant \citep{dvoretzky1956asymptotic,massart1990tight}, for target
    basket samples,
    \begin{equation}
        \Pr\left(
        \sup_b |\widehat F_p(b)-F_p(b)|>\epsilon
        \right)
        \le
        2e^{-2N\epsilon^2}.
        \label{eq:dkw_target}
    \end{equation}
    Equivalently, with probability at least $1-\delta$,
    \begin{equation}
        \sup_b |\widehat F_p(b)-F_p(b)|
        \le
        \sqrt{\frac{\log(2/\delta)}{2N}}.
        \label{eq:dkw_target_confidence}
    \end{equation}
    
    The same bound applies to the model empirical CDF $\widehat F_\phi$ with sample
    size $M$. Similar empirical-CDF concentration arguments have been used for
    finite-sample analysis of distributional risk measures such as VaR and CVaR
    \citep{kolla2019concentration}. Thus, for a fixed target distribution and a fixed circuit parameter $\phi$, the
    empirical CDFs converge uniformly in the basket value $b$ to their population
    counterparts. Consequently, the sample-estimated Basket-CDF loss provides a
    consistent estimate of the population CDF loss without reconstructing the full
    asset-register probability vector.
    
    The required target information consists of asset-wise marginals for the
    marginal loaders and samples, or an empirical CDF estimate, of the basket
    projection. Thus the
    target representation can scale as
    \[
        O(d2^q+J)
    \]
    for stored marginals and CDF grid values, or as
    \[
        O(d2^q+N)
    \]
    when the basket CDF is estimated from target samples, instead of requiring the full
    $O(2^{dq})$ joint probability table. In exact-grid numerical experiments, the reference CDF may still be computed from
    the full discretized joint law. The reduced scaling above refers to the training
    interface once marginal distributions and basket samples or CDF values are
    available.

    \section{State-Preparation Depth Analysis}
    \label{sec:complexity}
    
    This section analyzes the state-preparation circuit-depth scaling of the proposed Basket-CDF loader. We compare it with two reference loading methods: exact amplitude loading and Direct TT/MPS loading. We focus on state preparation, since this is the component that differs across the loading methods considered in this paper. Using the notation introduced in Section~\ref{subsec:multi_asset_loading}, the asset register contains $n=dq$ qubits.
    
    The expressions below describe theoretical depth scaling rather than exact hardware-specific compiled depths. Actual transpiled depths and CX counts are reported in the numerical experiments. A generic exact amplitude
    loader for a dense $n$-qubit state requires exponentially many elementary
    operations, so we write
    \[
        D_{\mathrm{exact}}
        =
        O(2^n)
        =
        O(2^{dq}).
    \]
    This is the state-preparation bottleneck that structure-aware loading methods
    aim to avoid.
    
    Direct TT/MPS loading provides a compressed alternative when the amplitude
tensor admits a low-rank tensor-network approximation. A bond-dimension-$\chi$
MPS admits a sequential-preparation interpretation using $O(n)$ local synthesis
steps, where $n=dq$ is the number of asset-register qubits
\citep{iaconis2024quantum}. In our numerical baseline, following the
MPS-to-circuit encoding perspective of \citet{ran2020encoding}, the target
amplitude tensor is first approximated by a rank-capped TT/MPS representation
and then converted into a circuit of one- and two-qubit gates. Since the compiled
depth depends on the retained bond dimension, local synthesis routine, and
compiler choices, we summarize the Direct TT state-preparation depth scaling as
\[
    D_{\mathrm{DirectTT}}
    =
    O\!\left(dq\,D_{\mathrm{loc}}(\chi)\right),
\]
where $D_{\mathrm{loc}}(\chi)$ denotes the implementation-dependent local
synthesis depth associated with a bond-dimension-$\chi$ MPS operation. In the
numerical experiments, we therefore report the actual transpiled depth and CX
count of the constructed Direct TT/MPS circuits.
    
    \subsection{Basket-CDF depth estimate}
    \label{subsec:basket_cdf_depth}
    
    The Basket-CDF loader consists of two stages. First, the asset-wise marginal
    distributions are prepared by local TT-informed marginal circuits. Second, a
    latent dependence block modifies the product-marginal state to match the basket
    pushforward distribution.
    
    For asset $a$, let $E_{\mathrm{TT}}^{(a)}$ be the TT-selected entangling edge set
    of the $q$-qubit marginal circuit, and let $L_{\mathrm{marg}}$ be the number of
    marginal-circuit layers. By the TT-informed ansatz construction in
    Section~\ref{subsec:tt_informed_hea}, the marginal loader for asset $a$ contains
    \begin{equation}
        N^{\mathrm{marg},a}_{2q}
        =
        L_{\mathrm{marg}}|E_{\mathrm{TT}}^{(a)}|
        \le
        L_{\mathrm{marg}}(q-1)
        \label{eq:marg_2q_count}
    \end{equation}
    two-qubit gates and
    \begin{equation}
        N^{\mathrm{marg},a}_{1q}
        =
        L_{\mathrm{marg}}q
        \label{eq:marg_1q_count}
    \end{equation}
    single-qubit rotations. If $\delta(E_{\mathrm{TT}}^{(a)})$ denotes the schedule
    depth of one marginal entangling block, then
    \begin{equation}
        D_a^{\mathrm{marg}}
        =
        O\!\left(
        L_{\mathrm{marg}}
        \left(1+\delta(E_{\mathrm{TT}}^{(a)})\right)
        \right).
        \label{eq:single_marg_depth}
    \end{equation}
    Under a sequential edge schedule,
    $\delta(E_{\mathrm{TT}}^{(a)})\le |E_{\mathrm{TT}}^{(a)}|\le q-1$, and hence
    \begin{equation}
        D_a^{\mathrm{marg}}
        =
        O(L_{\mathrm{marg}}q).
        \label{eq:single_marg_depth_bound}
    \end{equation}
    
    Since the marginal loaders act on disjoint asset registers, they can be
    scheduled in parallel at the logical-circuit level. Therefore the depth of the
    product-marginal initialization is
    \begin{equation}
        D_{\mathrm{marg}}
        =
        \max_{a=1,\ldots,d}
        D_a^{\mathrm{marg}}.
        \label{eq:dmarg_parallel}
    \end{equation}
    Using \eqref{eq:single_marg_depth_bound}, this gives
    \[
        D_{\mathrm{marg}}
        =
        O(L_{\mathrm{marg}}q).
    \]
    A fully serial implementation would replace the maximum in
    \eqref{eq:dmarg_parallel} by a sum, but the key point is that the marginal stage
    acts only on local $q$-qubit asset registers and does not require direct
    preparation of the full $dq$-qubit joint distribution.
    
    The dependence block used in the numerical implementation has one latent qubit
    unless stated otherwise. Let $\ell_0$ denote the latent qubit and let $q_{a,b}$
    denote the $b$th qubit of asset $a$. In our implementation, one dependence layer
    has the form
    \begin{equation}
        U_{\mathrm{dep}}^{(\ell)}
        =
        U_{\mathrm{rot},A}^{(\ell)}
        U_{\mathrm{chain}}^{(\ell)}
        U_{\mathrm{CRY}}^{(\ell)}
        R_y^{(\ell_0)} .
        \label{eq:dep_layer_form}
    \end{equation}
    Here
    \begin{equation}
        U_{\mathrm{CRY}}^{(\ell)}
        =
        \prod_{a=1}^{d}
        \prod_{b=0}^{q-1}
        \mathrm{CRY}_{\ell_0\rightarrow q_{a,b}} .
        \label{eq:cry_sweep}
    \end{equation}
    and
    \begin{equation}
        U_{\mathrm{chain}}^{(\ell)}
        =
        \prod_{a=1}^{d}
        \prod_{b=0}^{q-2}
        \mathrm{CNOT}_{q_{a,b}\rightarrow q_{a,b+1}},
        \qquad
        U_{\mathrm{rot},A}^{(\ell)}
        =
        \prod_{a=1}^{d}
        \prod_{b=0}^{q-1}
        R_y^{(q_{a,b})}.
        \label{eq:local_chain_rot}
    \end{equation}
    
    Each dependence layer contains latent-controlled rotations from the shared
    latent qubit to all asset qubits. Therefore, the number of latent-controlled
    rotations is
    \begin{equation}
        N^{\mathrm{dep}}_{\mathrm{CRY}}
        =
        L_{\mathrm{dep}}dq.
        \label{eq:dep_cry_count}
    \end{equation}
    The intra-asset CNOT chains contribute
    \begin{equation}
        N^{\mathrm{dep}}_{\mathrm{chain}}
        =
        L_{\mathrm{dep}}d(q-1)
        \label{eq:dep_chain_count}
    \end{equation}
    CNOT gates, and the one-qubit rotations contribute
    \begin{equation}
        N^{\mathrm{dep}}_{1q}
        =
        L_{\mathrm{dep}}(dq+1).
        \label{eq:dep_1q_count}
    \end{equation}
    
    Since the latent-controlled rotations share the same latent control qubit, they
    cannot be scheduled fully in parallel in the single-latent-qubit implementation.
    Their depth is therefore linear in $dq$ per dependence layer. The intra-asset
    CNOT chains add at most linear depth in $q$ per layer and can be executed across
    asset registers in parallel under ideal connectivity. Thus,
    \begin{equation}
        D_{\mathrm{dep}}
        =
        O(L_{\mathrm{dep}}dq).
        \label{eq:dep_depth}
    \end{equation}
    
    Combining the marginal initialization and the latent dependence block gives the
    online state-preparation depth
    \begin{equation}
        D_{\mathrm{CDF}}
        =
        D_{\mathrm{marg}}
        +
        O(L_{\mathrm{dep}}dq),
        \label{eq:cdf_depth}
    \end{equation}
    up to hardware-specific routing overhead and constant factors from native-gate
    decomposition. Since $D_{\mathrm{marg}}=O(L_{\mathrm{marg}}q)$ at the
    logical-circuit level, the dominant dependence on the number of assets comes
    from the latent dependence block. For fixed $L_{\mathrm{marg}}$,
    $L_{\mathrm{dep}}$, and $q$, the online state-preparation depth therefore grows
    linearly with $d$, rather than exponentially in the full joint grid size
    $2^{dq}$.
    
    \subsection{Summary of depth regimes}
    \label{subsec:depth_summary}
    
    The loading-depth regimes discussed above can be summarized in Table \ref{tab:depth_scaling}.
    
\begin{table}[htbp]
\centering
\caption{Summary of state-preparation depth scaling}
\label{tab:depth_scaling}
\begin{tabular}{lll}
\toprule
Construction & Relevant parameters & Depth scaling \\
\midrule
Exact loading
& $n=dq$
& $O(2^{dq})$ \\

Direct TT/MPS loading
& $\chi=\max_k r_k$
& $O(dq\,D_{\mathrm{loc}}(\chi))$ \\

TT-informed marginal primitive
& $L_{\mathrm{marg}},q,|E_{\mathrm{TT}}^{(a)}|$
& $O\!\left(L_{\mathrm{marg}}
\left(1+\delta(E_{\mathrm{TT}}^{(a)})\right)\right)$ \\

Basket-CDF loader
& $L_{\mathrm{marg}},L_{\mathrm{dep}},d,q$
& $D_{\mathrm{marg}}+O(L_{\mathrm{dep}}dq)$ \\
\bottomrule
\end{tabular}
\end{table}
    
    Here $D_{\mathrm{loc}}(\chi)$ denotes the local synthesis depth for a
    bond-dimension-$\chi$ MPS operation, and
    $\delta(E_{\mathrm{TT}}^{(a)})$ denotes the schedule depth of one TT-informed
    entangling block for the marginal circuit of asset $a$. These expressions describe theoretical loading-depth scaling rather than hardware-specific runtime estimates.
    
    The proposed framework does not claim a universal depth reduction for all
    multi-asset distributions. Direct TT/MPS loading can be highly accurate when a
    sufficiently large bond dimension is retained, but the corresponding local
    synthesis cost increases with $\chi$. The Basket-CDF loader takes a different
    route. It prepares asset-wise marginals locally and uses a linear-depth latent
    dependence block to match the basket pushforward distribution. This avoids
    full-joint state preparation in correlated basket-structured settings while
    keeping the online loading depth linear in $dq$ for fixed $L_{\mathrm{dep}}$ and
    fixed marginal-circuit depth.
    
    The numerical resource counts in Appendix~\ref{app:numerical} are consistent
    with this estimate. After transpilation, the Basket-CDF depth grows with an
    approximately constant per-asset increment. In the $q=3$ experiments, for
    instance, adding one asset increases the Basket-CDF depth by about $36$ circuit
    layers, while exact amplitude loading grows by orders of magnitude over the same
    range. Thus, the linear-depth dependence in \eqref{eq:cdf_depth} remains visible
    in the transpiled circuits, with a smaller constant factor due to the shallow
    latent-dependence block.

     \section{Numerical Results}
    \label{sec:numerical_results}
    
    We evaluate the proposed loading methods on discretized multi-asset lognormal
    basket instances. The numerical experiments are designed to isolate the effect of
    state preparation on basket-structured payoff accuracy. Accordingly, all option
    values reported in this section are computed directly from the probability
    distribution induced by each state-preparation circuit. Quantum amplitude estimation
    and payoff-oracle sampling are not included in the reported errors.
    
    The main comparison is between a Direct TT/MPS loader with a fixed rank cap and the
    proposed marginal-TT latent Basket-CDF loader. The Direct TT baseline first constructs
    a rank-capped tensor-train representation of the target amplitude tensor and then
    converts it into a quantum circuit. The proposed loader first trains asset-wise
    TT-informed marginal circuits, appends one latent qubit, and trains a latent
    dependence block using a Basket-CDF loss with marginal regularization. The proposed
    training objective does not use option-price labels or full-state fidelity.
    
    For each market instance, we compare the induced basket payoff values against the
    exact discretized reference distribution at a fixed set of strikes. We report payoff errors, full-distribution fidelity as a diagnostic, transpiled
    state-preparation depth, and CX count. The software environment is summarized in Appendix~\ref{app:software_environment},
    and the market-sector-style correlation construction is described in
    Appendix~\ref{app:market_sector_corr}. Market parameters, discretization choices,
    and optimization settings are specified below. The Direct TT baseline is evaluated with a fixed maximum bond dimension
    $\chi=4$. This cap was chosen to provide a compressed low-rank TT/MPS baseline
    with moderate circuit depth. Increasing $\chi$ may improve the Direct TT
    approximation, but it also increases the local MPS synthesis cost and therefore
    leads to deeper online loading circuits.

    \subsection{Latent-register Basket-CDF loader in correlated regimes}
    \label{subsec:latent_basket_cdf_results}
    
    As a preliminary validation of the TT-informed state-preparation primitive,
    Appendix~\ref{app:tt_independent} reports the independent case, where the joint
    distribution factorizes across asset blocks. In this regime, the TT-rank profile
    removes unnecessary cross-asset entanglers and yields a shallow full-register
    state-preparation circuit. The main numerical focus below is the correlated
    basket setting, where TT-informed circuits are used locally for marginal
    preparation and the Basket-CDF latent block learns basket-relevant dependence.
    
   We next evaluate the marginal-TT latent Basket-CDF loader in correlated basket
settings. Figures~\ref{fig:q2_directtt_ours_comparison}
and~\ref{fig:q3_directtt_ours_comparison} report the comparison for $q=2$ and
$q=3$ uncertainty qubits per asset, respectively. Both figures show the pricing error at $K=55$, near the central payoff-sensitive region for $S_0=50$,
$r=0.05$, and $T=1$. Full strike-wise results are reported in
Appendix~\ref{app:numerical}.
    
    For accuracy, the proposed Basket-CDF loader is competitive with or better than
    the fixed-rank Direct TT baseline across the correlated tests, with the
    advantage most visible in the market-sector-style correlation regime. In the
    equicorrelation case, Direct TT can perform well at low dimension, indicating
    that this correlation structure can still be reasonably compressible by a
    low-rank TT representation with the chosen bond-dimension cap $\chi=4$. However,
    as the dimension increases, the Basket-CDF loader remains more stable. In the
    market-sector-style case, the difference is clearer: Direct TT becomes more
    sensitive to dimension and correlation heterogeneity, whereas the Basket-CDF
    loader maintains lower pricing errors. This suggests that Direct TT is effective when the full
    joint distribution is well aligned with a low-rank TT structure, while the
    proposed method is more robust when the payoff-relevant basket distribution is
    easier to learn than the full joint law.
    
    The depth comparison shows the main resource advantage of the proposed method.
    Exact amplitude loading grows exponentially with the total number of uncertainty
    qubits, and therefore its transpiled depth increases rapidly with $d$ and $q$.
    By contrast, both Direct TT with fixed $\chi$ and the proposed Basket-CDF loader
    grow approximately linearly in the tested regimes. Among the two compressed loaders, the new Basket-CDF circuit is consistently
    shallower in the higher-resolution tests and has a smaller depth growth with
    dimension. For $q=3$, its depth increases from $92$ to $200$ as
    $d$ increases from $2$ to $5$, whereas the fixed-rank Direct TT depth increases
    from $133$ to $482$ over the same range.
    
    The higher-resolution $q=4$ results in Table~\ref{tab:q4-full-results} further
    support this interpretation. Although the $q=4$ experiments are limited to
    $d\le 4$ because of the offline cost of exact classical training and validation,
    they show that the proposed Basket-CDF loader becomes more stable as the
    uncertainty-qubit resolution is increased. For example, at $d=4$ and $K=55$, the Basket-CDF error decreases from
    $1.363\%$ for $q=2$ to $0.343\%$ for $q=3$ and $0.173\%$ for $q=4$ in the
    equicorrelation case. In the market-sector-style case, the corresponding errors
    decrease from $1.681\%$ to $0.211\%$ and then to $0.156\%$. This behavior is consistent with
    Remark~\ref{rem:resolution_error_decomposition}: increasing $q$ refines the
    marginal asset grids and the induced basket distribution, while the Basket-CDF
    objective continues to control the loader-induced error through the basket CDF.
    
    Overall, these results show that the proposed loader trades offline variational
    training for a shallow online state-preparation circuit. The resulting circuit
    avoids the exponential depth growth of exact loading and is less dependent than
    Direct TT on whether the full joint amplitude tensor is compressible at a fixed
    bond dimension. This makes the Basket-CDF loader particularly attractive for
    correlated basket-pricing problems, where the target quantity depends on the
    basket pushforward distribution rather than on full-state fidelity.
    
    We also performed two additional robustness analyses, reported in Appendix~\ref{app:numerical}. First, Table~\ref{tab:sampletarget_basketcdf_q3_mean} replaces the exact-grid Basket-CDF target with a sample-estimated basket CDF. The resulting errors remain at the percent level across the tested dimensions and correlation structures. This supports the sample-estimable formulation in Section~\ref{subsec:sample_estimable}, showing that the proposed loader can be trained from samples of the basket projection rather than from the full joint probability tensor.
    
    Second, Table~\ref{tab:directtt_rank_sweep_q3} reports a bond-dimension sweep for the Direct TT/MPS baseline. As expected, increasing the TT bond-dimension cap improves the approximation accuracy and drives the fidelity closer to one. However, this improvement comes with a substantial increase in transpiled state-preparation depth and CX count. This confirms that Direct TT/MPS loading can be highly accurate when a sufficiently large bond dimension is retained, but the additional accuracy is obtained at the cost of deeper online loading circuits. Therefore, the fixed-rank Direct TT baseline should be interpreted as a
    compressed-loading baseline at a moderate bond dimension rather than as the best
    possible Direct TT approximation.
    
    Finally, Figure~\ref{fig:qae_sanity_check} evaluates the compatibility of the
    sampling-based Basket-CDF loader with a standard QAE-based pricing workflow
    for one representative correlated basket instance. 
    The QAE estimates
    closely follow the option-price curve induced by the sampling-based Basket-CDF loader
    and remain close to the exact discretized reference across the tested strikes.
    
    As an additional empirical-scenario robustness, Appendix~\ref{app:empirical_scenario}
    applies the same Basket-CDF construction to historical equity return scenarios.
    This experiment is not intended as a risk-neutral pricing calibration; rather,
    it tests whether the loader can match an empirical basket distribution when the
    target is specified by samples instead of a closed-form probability model.
    
    \begin{figure}[t]
        \centering
        \begin{subfigure}[t]{0.48\textwidth}
            \centering
            \includegraphics[width=\textwidth]{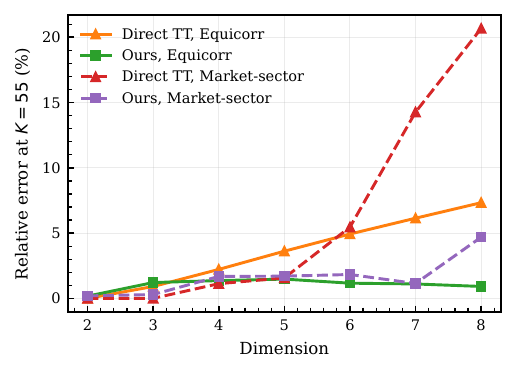}
            \caption{Relative pricing error at $K=55$.}
            \label{fig:q2_comparison_error_k55}
        \end{subfigure}
        \hfill
        \begin{subfigure}[t]{0.48\textwidth}
            \centering
            \includegraphics[width=\textwidth]{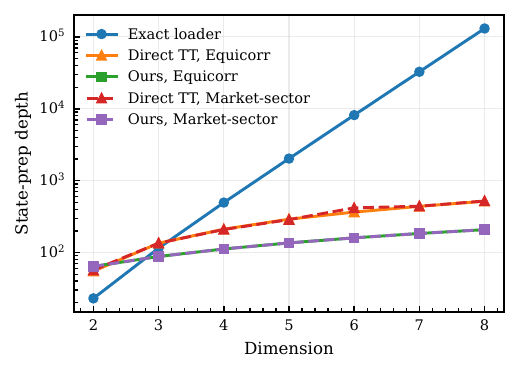}
            \caption{State-preparation depth.}
            \label{fig:q2_comparison_depth}
        \end{subfigure}
        \caption{Correlated-basket comparison for $q=2$ uncertainty qubits per asset. (a) Relative pricing error at strike $K=55$. (b) Transpiled state-preparation depth on a logarithmic scale. Exact loading is shown only in panel (b), since it prepares the discretized target state exactly. See Table \ref{tab:directtt_basketcdf_q2_mean} for the full results.}
        \label{fig:q2_directtt_ours_comparison}
    \end{figure}

    \begin{figure}[t]
        \centering
        \begin{subfigure}[t]{0.48\textwidth}
            \centering
            \includegraphics[width=\textwidth]{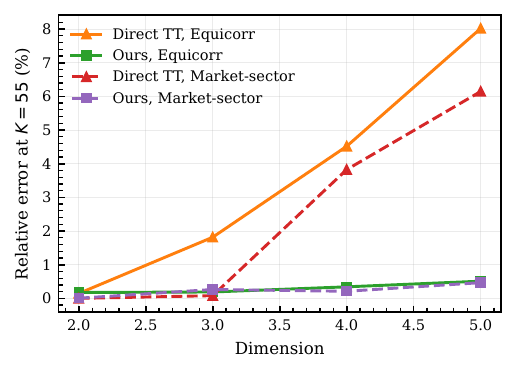}
            \caption{Relative pricing error at $K=55$.}
            
        \end{subfigure}
        \hfill
        \begin{subfigure}[t]{0.48\textwidth}
            \centering
            \includegraphics[width=\textwidth]{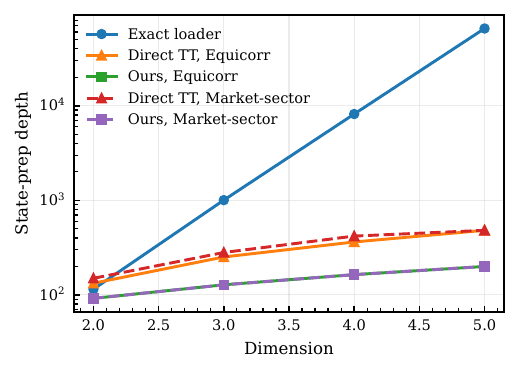}
            \caption{State-preparation depth.}
            
        \end{subfigure}
        \caption{Correlated-basket comparison for $q = 3$ uncertainty qubits per asset.
(a) Relative pricing error at K = 55. (b) Transpiled state-preparation depth on a
logarithmic scale. Exact loading is shown only in panel (b), since it prepares the
discretized target state exactly. See Table \ref{tab:directtt_basketcdf_q3_mean} for the full results.}
        \label{fig:q3_directtt_ours_comparison}
    \end{figure}

\section{Discussion and Conclusion}
\label{sec:discuss}

This paper addresses the state-preparation cost in QAE-based basket option
pricing. Although QAE can reduce the sampling complexity of Monte Carlo
estimation, its practical benefit critically depends on repeatedly preparing the
discretized asset-price distribution as a quantum input state. We proposed a
structure-aware state-preparation framework that combines TT-informed
asset-wise marginal loaders with a latent dependence block trained
through a Basket-CDF objective. The resulting loader is designed to preserve
the basket-value distribution rather than to reconstruct the full joint
asset distribution.

The numerical results show that this task-aligned construction can retain
basket-level pricing accuracy while substantially reducing online
state-preparation depth. Exact amplitude loading has exponential depth
scaling in the full asset register, whereas the proposed Basket-CDF loader has
linear depth scaling in $dq$ for fixed marginal and dependence-block depths.
Across the correlated basket experiments, the proposed loader maintains
low-percent pricing errors and remains competitive with, or more stable than,
the fixed-rank Direct TT/MPS baseline. The Direct TT/MPS bond-dimension sweep in
Table~\ref{tab:directtt_rank_sweep_q3} further illustrates the
accuracy--depth trade-off: increasing the TT bond-dimension cap improves
full-state fidelity, but this improvement comes at the cost of deeper online
state-preparation circuits and larger CX counts. Thus, the proposed method should not be interpreted as universally more
accurate than high-rank tensor-network loading. Its advantage is that it
preserves basket-pricing accuracy using a shallow, task-aligned circuit.

A key reason this approximation can work well is the structure of basket
payoffs. Basket option values depend on the multi-asset distribution only
through the distribution of the weighted basket value. Therefore, accurate basket pricing does not require perfect reconstruction of
the full joint distribution, provided that the basket pushforward distribution
is well preserved. This is consistent with the Basket-CDF analysis in
Section~\ref{subsec:basket_cdf_objective}, where basket-call errors are
controlled by the CDF discrepancy of the basket value rather than by
full-state fidelity.

Figure~\ref{fig:qae_sanity_check} further supports this interpretation through
a QAE integration experiment using a Basket-CDF loader obtained from
sampling-based training. When the learned loader is inserted into a standard
QAE-based pricing workflow, the QAE estimates track the option-price curve
induced by the loader and remain close to the exact discretized reference
across the tested strikes.

The current implementation uses exact statevector simulation for controlled training and validation, which becomes expensive as
$dq$ increases. This exact-simulation workflow is useful for controlled
benchmarking, while the Basket-CDF objective itself is sample-estimable. The
target basket CDF can be estimated from samples of the basket projection, and
the model CDF can be estimated from circuit measurement samples without
reconstructing the full asset-register probability vector. The
sampling-based Basket-CDF training results in
Table~\ref{tab:sampletarget_basketcdf_q3_mean} show that this
sample-estimated training interface can still achieve percent-level pricing
errors in the tested cases. 

Future work will extend the sampling-based training interface developed in
this paper to larger asset registers and more realistic execution settings.
Natural directions include more efficient sampling-based training,
noisy-simulator and real-device experiments, adaptive latent-register designs,
and extensions to other portfolio-level quantities such as tail probabilities,
value-at-risk (VaR), and conditional value-at-risk (CVaR).

\backmatter

\section*{Declarations}

\subsection*{Availability of data and materials}
The datasets generated and/or analyzed during the current study are available from the corresponding author on reasonable request.

\subsection*{Competing interests}
The authors declare that they have no competing interests.

\subsection*{Funding}
This work was supported by the National Research Foundation of Korea
(RS-2025-02309510, RS-2025-16066836).

\subsection*{Authors' contributions}
DK developed the methodology, implemented the numerical experiments, analyzed the results, and drafted the manuscript. ZC contributed to manuscript revision and helped assess the financial relevance and interpretation of the numerical results. DKP contributed to methodology development and manuscript revision. CL supervised the project, contributed to conceptualization and methodology, and revised the manuscript. All authors contributed to the discussion of the results and read and approved the final manuscript.

\subsection*{Acknowledgements}
Not applicable.






\begin{appendices}

    \section{Software Environment}
    \label{app:software_environment}

All numerical experiments were implemented in Python.
The implementation does not rely on local project-specific modules. The
experiments were run in a Python environment compatible with the
\texttt{mps-to-circuit} package; in particular, the notebook uses a Python
version in the range
\[
    3.10 \le \texttt{Python} < 3.13,
\]
with Python 3.12 used as the reference environment. The reference environment used in our experiments was
\[
\begin{aligned}
&\texttt{Python}=3.12.13,\quad
\texttt{numpy}=2.4.4,\quad
\texttt{pandas}=3.0.2,\\
&\texttt{scipy}=1.17.1,\quad
\texttt{matplotlib}=3.10.9,\quad
\texttt{qiskit}=2.4.1,\\
&\texttt{qiskit-aer}=0.17.2,\quad
\texttt{mps-to-circuit}=0.1.2.
\end{aligned}
\]

The package \texttt{numpy} is used for tensor and probability-vector operations,
\texttt{pandas} for result aggregation, and \texttt{scipy} for optimization and
multivariate normal density evaluation. Circuit construction and transpilation
are performed with \texttt{qiskit}, while exact statevector probabilities are
obtained using \texttt{qiskit.quantum\_info.Statevector}. The Direct TT/MPS
baseline is constructed using the \texttt{mps-to-circuit} package.

All variational optimizations use \texttt{scipy.optimize.minimize} with the
L-BFGS-B method. Multi-seed experiments use explicitly specified pseudorandom
seeds, and reported multi-seed quantities are computed by aggregating the
corresponding per-seed results.

All reported circuit resources are computed after transpilation to the basis
\[
    \{\texttt{u},\texttt{cx}\}
\]
with Qiskit's transpiler. Unless otherwise stated, no backend-specific device
topology, routing constraint, or hardware noise model is included in the reported
resource counts. The reported pricing and fidelity diagnostics are computed from
classical statevector simulations and therefore represent state-preparation
accuracy before finite-shot quantum estimation error.

    \section{Market-sector-style Correlation Structure}
    \label{app:market_sector_corr}
    
    This appendix describes the market-sector-style correlation matrix used in the
    correlated numerical experiments. The purpose of this case is to provide a
    structured but non-equicorrelated dependence pattern. The construction is a synthetic multi-factor correlation model inspired by equity
    factor-risk models, where asset returns are decomposed into common systematic factors
    and idiosyncratic residual risk
    \citep{rosenberg1973prediction,fama1993common,carhart1997persistence}.\footnote{
    Commercial equity risk models, such as MSCI Barra models, also use market, industry,
    style, and other common risk factors together with asset-specific residual risk. See,
    for example, MSCI's Barra Global Equity Model documentation:
    \url{https://www.msci.com/documents/10199/242721/Barra_Global_Equity_Model_GEM3.pdf}.
    We use the term ``market-sector-style'' only to indicate this modeling motivation; the
    benchmark in this paper is not a calibrated implementation of any commercial Barra
    model.
    }
    The model should therefore be interpreted as a reproducible synthetic factor
    correlation structure rather than as a historically calibrated market model. It is
    designed to include market-wide dependence, sector clustering, style-like
    cross-sectional variation, and idiosyncratic risk.
    
    For a problem with \(d\) assets, index the assets by
    \[
        a=0,\ldots,d-1.
    \]
    Define a normalized cross-sectional coordinate
    \[
        z_a
        =
        \mathrm{linspace}(-1,1,d)_a .
    \]
    Each asset is assigned a five-dimensional loading vector
    \[
        \ell_a
        =
        \left(
            \ell_a^{\mathrm{mkt}},
            \ell_a^{\mathrm{secA}},
            \ell_a^{\mathrm{secB}},
            \ell_a^{\mathrm{val}},
            \ell_a^{\mathrm{mom}}
        \right),
    \]
    where
    \[
        \ell_a^{\mathrm{mkt}} = 0.42,
    \]
    \[
        \ell_a^{\mathrm{secA}}
        =
        \begin{cases}
            0.48, & a < \lceil d/2\rceil,\\
            0.03, & a \ge \lceil d/2\rceil,
        \end{cases}
    \]
    \[
        \ell_a^{\mathrm{secB}}
        =
        \begin{cases}
            0.03, & a < \lceil d/2\rceil,\\
            0.47, & a \ge \lceil d/2\rceil,
        \end{cases}
    \]
    and
    \[
        \ell_a^{\mathrm{val}} = 0.38 z_a,
        \qquad
        \ell_a^{\mathrm{mom}} = 0.26(-1)^a.
    \]
    The market component creates a common exposure shared by all assets. The two sector
    components create stronger within-group dependence in the first and second halves of
    the asset universe. The value-like component introduces a smooth cross-sectional
    gradient, while the momentum-like component introduces an alternating style exposure
    across asset indices.
    
    Let \(L\in\mathbb{R}^{d\times 5}\) be the factor loading matrix whose \(a\)th row is
    \(\ell_a\). We form the covariance proxy
    \[
        C
        =
        LL^\top + \eta I_d,
        \qquad
        \eta=0.22,
    \]
    where \(\eta I_d\) represents idiosyncratic variance. The correlation matrix is then
    obtained by normalizing this covariance proxy:
    \[
        \rho_{ij}
        =
        \frac{C_{ij}}{\sqrt{C_{ii}C_{jj}}},
        \qquad
        i,j=0,\ldots,d-1,
    \]
    with \(\rho_{ii}=1\).
    
    The matrix \(C\) is positive definite by construction because \(LL^\top\) is positive
    semidefinite and \(\eta I_d\) is positive definite for \(\eta>0\). The resulting
    correlation matrix is therefore valid and contains heterogeneous dependence beyond a
    single equicorrelation coefficient. This market-sector-style case is used to test the
    loading methods on a structured multi-factor correlation pattern rather than on a
    purely homogeneous correlation matrix.
    
    \section{TT-informed State Preparation in the Independent Case}
    \label{app:tt_independent}

    This appendix reports a preliminary validation of the TT-informed
    state-preparation primitive in the independent regime. When the assets are
    independent, the joint distribution factorizes across asset blocks, and the
    corresponding amplitude tensor has TT rank one across inter-asset cuts. The
    TT-informed entangling rule therefore removes unnecessary cross-asset entanglers
    and retains only the local links needed for marginal state preparation.
    
    Figure~\ref{fig:independent_regime} reports the results for $\rho=0$ and
    $q=2$ uncertainty qubits per asset. The exact loader has zero pricing error
    because it prepares the discretized target state directly. The TT-informed
    loader introduces only a small approximation error, remaining at the percent
    level across the tested dimensions. At the same time, it achieves a substantial
    depth reduction relative to exact amplitude loading. While exact loading grows
    rapidly with the total number of uncertainty qubits, the TT-informed circuit has
    nearly constant transpiled depth in this independent experiment, since adding
    assets does not require additional cross-asset entanglement.
    
    These results support the use of TT-rank information as a structural prior for
    removing redundant entangling links in factorized settings. In the main
    correlated-basket experiments, however, we do not use this full-register
    TT-informed circuit as the primary loader. Instead, the same TT-informed
    primitive is used locally for asset-wise marginal preparation, while
    cross-asset dependence is learned by the latent Basket-CDF block.
    
    \begin{figure}[t]
        \centering
        \begin{subfigure}[t]{0.48\textwidth}
            \centering
            \includegraphics[width=\textwidth]{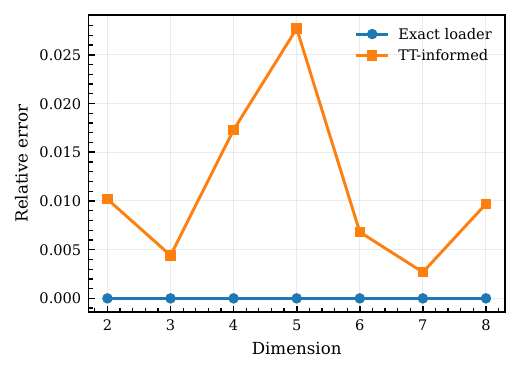}
            \caption{Relative pricing error vs.\ dimension.}
        \end{subfigure}
        \hfill
        \begin{subfigure}[t]{0.48\textwidth}
            \centering
            \includegraphics[width=\textwidth]{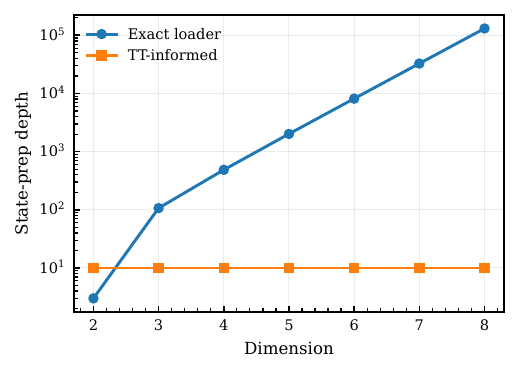}
            \caption{State-preparation depth vs.\ dimension.}
        \end{subfigure}
        \caption{
        Results for the independent regime $(\rho=0)$ with $q=2$ uncertainty qubits per
        asset. (a) Relative pricing error versus basket dimension. (b) Transpiled
        state-preparation depth versus basket dimension on a logarithmic scale.
        } \label{fig:independent_regime}
    \end{figure}
    
    \FloatBarrier 
\section{Additional Numerical Results}
\label{app:numerical}

\begin{table}[t]
\centering
\scriptsize
\setlength{\tabcolsep}{2pt}
\caption{Comparison between Exact loading, Direct TT, and the proposed Basket-CDF loader for $q=2$ uncertainty qubits per asset and dimensions $d=2,\ldots,8$. Market parameters are $S_0=50$, $\sigma=0.3$, $r=0.05$, and $T=1.0$. Results are shown for the equicorrelation case $\rho=0.6$ and the market-sector-style correlation case. Basket-CDF entries are multi-seed means over seeds $0,\ldots,9$ obtained with the loss weight $\lambda_B = 1000, \lambda_M = 10$, one latent qubit, $L_{\mathrm{dep}}=4$, and L-BFGS-B with 180 maximum iterations. Errors are relative pricing errors in percent.}
\label{tab:directtt_basketcdf_q2_mean}

\begin{tabular*}{\linewidth}{
    @{\extracolsep{\fill}}
    lllccccccc
    @{}
}
\toprule
Case & $d$ & Method
& $K=40$ & $K=45$ & $K=50$ & $K=55$
& Fidelity & Depth & CX \\
\midrule

\multirow{21}{*}{Equicorr ($\rho=0.6$)}
& \multirow{3}{*}{2}
& Exact loader
& -- & -- & -- & -- & -- & 23 & 11 \\

&
& Direct TT
& 0.000 & 0.000 & 0.000 & 0.000 & 1.0000 & 55 & 30 \\

&
& Basket-CDF
& 0.078 & 0.097 & 0.128 & 0.172 & 0.9900 & 64 & 50 \\

& \multirow{3}{*}{3}
& Exact loader
& -- & -- & -- & -- & -- & 115 & 57 \\

&
& Direct TT
& 0.512 & 0.594 & 0.716 & 0.910 & 0.9993 & 135 & 72 \\

&
& Basket-CDF
& 0.499 & 0.669 & 0.893 & 1.213 & 0.9378 & 88 & 75 \\

& \multirow{3}{*}{4}
& Exact loader
& -- & -- & -- & -- & -- & 495 & 247 \\

&
& Direct TT
& 1.358 & 1.575 & 1.854 & 2.220 & 0.9988 & 211 & 112 \\

&
& Basket-CDF
& 0.593 & 0.790 & 0.988 & 1.363 & 0.8880 & 112 & 100 \\

& \multirow{3}{*}{5}
& Exact loader
& -- & -- & -- & -- & -- & 2027 & 1013 \\

&
& Direct TT
& 2.463 & 2.828 & 3.190 & 3.615 & 0.9983 & 291 & 154 \\

&
& Basket-CDF
& 0.558 & 0.802 & 1.000 & 1.476 & 0.8548 & 136 & 125 \\

& \multirow{3}{*}{6}
& Exact loader
& -- & -- & -- & -- & -- & 8167 & 4083 \\

&
& Direct TT
& 3.738 & 4.195 & 4.487 & 4.929 & 0.9982 & 365 & 193 \\

&
& Basket-CDF
& 0.513 & 0.660 & 0.784 & 1.164 & 0.8991 & 160 & 150 \\

& \multirow{3}{*}{7}
& Exact loader
& -- & -- & -- & -- & -- & 32739 & 16369 \\

&
& Direct TT
& 5.075 & 5.516 & 5.710 & 6.142 & 0.9984 & 439 & 232 \\

&
& Basket-CDF
& 0.906 & 1.005 & 0.926 & 1.104 & 0.9351 & 184 & 175 \\

& \multirow{3}{*}{8}
& Exact loader
& -- & -- & -- & -- & -- & 131039 & 65519 \\

&
& Direct TT
& 6.486 & 6.845 & 6.976 & 7.328 & 0.9988 & 517 & 273 \\

&
& Basket-CDF
& 0.833 & 0.983 & 0.646 & 0.914 & 0.9614 & 208 & 200 \\

\midrule

\multirow{21}{*}{Market-sector}
& \multirow{3}{*}{2}
& Exact loader
& -- & -- & -- & -- & -- & 23 & 11 \\

&
& Direct TT
& 0.000 & 0.000 & 0.000 & 0.000 & 1.0000 & 57 & 31 \\

&
& Basket-CDF
& 0.071 & 0.080 & 0.103 & 0.191 & 0.9999 & 64 & 50 \\

& \multirow{3}{*}{3}
& Exact loader
& -- & -- & -- & -- & -- & 115 & 57 \\

&
& Direct TT
& 0.001 & 0.002 & 0.002 & 0.002 & 1.0000 & 135 & 72 \\

&
& Basket-CDF
& 0.089 & 0.131 & 0.195 & 0.296 & 0.9403 & 88 & 75 \\

& \multirow{3}{*}{4}
& Exact loader
& -- & -- & -- & -- & -- & 495 & 247 \\

&
& Direct TT
& 0.474 & 0.591 & 0.584 & 1.126 & 0.9986 & 209 & 111 \\

&
& Basket-CDF
& 0.429 & 0.649 & 0.972 & 1.681 & 0.7892 & 112 & 100 \\

& \multirow{3}{*}{5}
& Exact loader
& -- & -- & -- & -- & -- & 2027 & 1013 \\

&
& Direct TT
& 1.324 & 1.424 & 1.255 & 1.568 & 0.9973 & 287 & 152 \\

&
& Basket-CDF
& 0.464 & 0.682 & 1.071 & 1.702 & 0.7851 & 136 & 125 \\

& \multirow{3}{*}{6}
& Exact loader
& -- & -- & -- & -- & -- & 8167 & 4083 \\

&
& Direct TT
& 3.056 & 3.540 & 3.933 & 5.458 & 0.9950 & 418 & 203 \\

&
& Basket-CDF
& 0.391 & 0.462 & 0.986 & 1.829 & 0.7244 & 160 & 150 \\

& \multirow{3}{*}{7}
& Exact loader
& -- & -- & -- & -- & -- & 32739 & 16369 \\

&
& Direct TT
& 5.146 & 6.966 & 9.349 & 14.267 & 0.9939 & 441 & 233 \\

&
& Basket-CDF
& 0.297 & 0.589 & 0.855 & 1.140 & 0.7540 & 184 & 175 \\

& \multirow{3}{*}{8}
& Exact loader
& -- & -- & -- & -- & -- & 131039 & 65519 \\

&
& Direct TT
& 8.057 & 10.816 & 13.485 & 20.672 & 0.9920 & 521 & 275 \\

&
& Basket-CDF
& 0.922 & 1.374 & 2.527 & 4.686 & 0.8039 & 208 & 200 \\

\bottomrule
\end{tabular*}
\end{table}

\begin{table}[t]
\centering
\scriptsize
\setlength{\tabcolsep}{2pt}
\caption{Comparison between Exact loading, Direct TT, and the proposed Basket-CDF loader for $q=3$ uncertainty qubits per asset and dimensions $d=2,\ldots,5$. Market parameters are $S_0=50$, $\sigma=0.3$, $r=0.05$, and $T=1.0$. Results are shown for the equicorrelation case $\rho=0.6$ and the market-sector-style correlation case. Basket-CDF entries are multi-seed means over seeds $0,\ldots,9$ obtained with the loss weight $\lambda_B = 1000, \lambda_M = 10$, one latent qubit, $L_{\mathrm{dep}}=4$, and L-BFGS-B with 180 maximum iterations. Errors are relative pricing errors in percent.}
\label{tab:directtt_basketcdf_q3_mean}

\begin{tabular*}{\linewidth}{
    @{\extracolsep{\fill}}
    lllccccccc
    @{}
}
\toprule
Case & $d$ & Method
& $K=40$ & $K=45$ & $K=50$ & $K=55$
& Fidelity & Depth & CX \\
\midrule

\multirow{12}{*}{Equicorr ($\rho=0.6$)}
& \multirow{3}{*}{2}
& Exact loader
& -- & -- & -- & -- & -- & 115 & 57 \\

&
& Direct TT
& 0.030 & 0.019 & 0.058 & 0.146 & 0.9999 & 133 & 71 \\

&
& Basket-CDF
& 0.047 & 0.062 & 0.105 & 0.175 & 0.8571 & 92 & 84 \\

& \multirow{3}{*}{3}
& Exact loader
& -- & -- & -- & -- & -- & 1005 & 502 \\

&
& Direct TT
& 0.208 & 0.316 & 0.997 & 1.816 & 0.9943 & 252 & 132 \\

&
& Basket-CDF
& 0.054 & 0.062 & 0.133 & 0.190 & 0.6800 & 128 & 126 \\

& \multirow{3}{*}{4}
& Exact loader
& -- & -- & -- & -- & -- & 8167 & 4083 \\

&
& Direct TT
& 0.698 & 1.091 & 2.625 & 4.517 & 0.9846 & 363 & 192 \\

&
& Basket-CDF
& 0.159 & 0.178 & 0.245 & 0.343 & 0.5435 & 164 & 168 \\

& \multirow{3}{*}{5}
& Exact loader
& -- & -- & -- & -- & -- & 65505 & 32752 \\

&
& Direct TT
& 1.332 & 2.124 & 4.761 & 8.016 & 0.9701 & 482 & 253 \\

&
& Basket-CDF
& 0.210 & 0.240 & 0.319 & 0.511 & 0.4169 & 200 & 210 \\

\midrule

\multirow{12}{*}{Market-sector}
& \multirow{3}{*}{2}
& Exact loader
& -- & -- & -- & -- & -- & 115 & 57 \\

&
& Direct TT
& 0.000 & 0.000 & 0.000 & 0.000 & 1.0000 & 150 & 74 \\

&
& Basket-CDF
& 0.001 & 0.003 & 0.005 & 0.005 & 0.9999 & 92 & 84 \\

& \multirow{3}{*}{3}
& Exact loader
& -- & -- & -- & -- & -- & 1005 & 502 \\

&
& Direct TT
& 0.025 & 0.019 & 0.035 & 0.081 & 0.9998 & 281 & 136 \\

&
& Basket-CDF
& 0.056 & 0.080 & 0.131 & 0.264 & 0.8599 & 128 & 126 \\

& \multirow{3}{*}{4}
& Exact loader
& -- & -- & -- & -- & -- & 8167 & 4083 \\

&
& Direct TT
& 0.502 & 1.273 & 2.256 & 3.824 & 0.9896 & 419 & 203 \\

&
& Basket-CDF
& 0.041 & 0.167 & 0.199 & 0.211 & 0.6777 & 164 & 168 \\

& \multirow{3}{*}{5}
& Exact loader
& -- & -- & -- & -- & -- & 65505 & 32752 \\

&
& Direct TT
& 0.838 & 1.816 & 3.371 & 6.145 & 0.9613 & 482 & 253 \\

&
& Basket-CDF
& 0.107 & 0.287 & 0.364 & 0.467 & 0.4979 & 200 & 210 \\

\bottomrule
\end{tabular*}
\end{table}

\begin{table}[t]
\centering
\scriptsize
\setlength{\tabcolsep}{2pt}
\caption{Comparison between Exact loading, Direct TT, and the proposed Basket-CDF loader for $q = 4$ uncertainty qubits per asset and dimensions $d = 2,\ldots,4$. Market parameters are $S_0 = 50$, $\sigma = 0.3$, $r = 0.05$, and $T = 1.0$. Results are shown for the equicorrelation case $\rho = 0.6$ and the market-sector-style correlation case. Basket-CDF entries are multi-seed means over seeds $0,\ldots,9$ obtained with the loss weights $\lambda_B = 1000$ and $\lambda_M = 10$, one latent qubit, $L_{\mathrm{dep}} = 4$, and L-BFGS-B with 180 maximum iterations. Errors are relative pricing errors in percent.}
\label{tab:q4-full-results}

\begin{tabular*}{\linewidth}{
    @{\extracolsep{\fill}}
    lllrrrrrrr
    @{}
}
\toprule
Case & $d$ & Method
& $K=40$ & $K=45$ & $K=50$ & $K=55$
& Fidelity & Depth & CX \\
\midrule

\multirow{9}{*}{Equicorr ($\rho = 0.6$)}
& \multirow{3}{*}{2}
& Exact loader
& -- & -- & -- & -- & -- & 495 & 247 \\

&
& Direct TT
& 0.025 & 0.017 & 0.185 & 0.264 & 0.9993 & 213 & 113 \\

&
& Basket-CDF
& 0.040 & 0.030 & 0.052 & 0.127 & 0.8176 & 116 & 118 \\

& \multirow{3}{*}{3}
& Exact loader
& -- & -- & -- & -- & -- & 8167 & 4083 \\

&
& Direct TT
& 0.526 & 1.043 & 2.138 & 3.793 & 0.9849 & 363 & 192 \\

&
& Basket-CDF
& 0.037 & 0.078 & 0.074 & 0.135 & 0.6136 & 164 & 177 \\

& \multirow{3}{*}{4}
& Exact loader
& -- & -- & -- & -- & -- & 131039 & 65519 \\

&
& Direct TT
& 0.394 & 0.413 & 2.302 & 4.589 & 0.9640 & 523 & 276 \\

&
& Basket-CDF
& 0.056 & 0.155 & 0.195 & 0.173 & 0.4341 & 212 & 236 \\

\midrule

\multirow{9}{*}{Market-sector}
& \multirow{3}{*}{2}
& Exact loader
& -- & -- & -- & -- & -- & 495 & 247 \\

&
& Direct TT
& 0.000 & 0.000 & 0.000 & 0.000 & 1.0000 & 213 & 113 \\

&
& Basket-CDF
& 0.001 & 0.001 & 0.001 & 0.003 & 0.9993 & 116 & 118 \\

& \multirow{3}{*}{3}
& Exact loader
& -- & -- & -- & -- & -- & 8167 & 4083 \\

&
& Direct TT
& 0.027 & 0.043 & 0.165 & 0.375 & 0.9995 & 367 & 194 \\

&
& Basket-CDF
& 0.064 & 0.037 & 0.092 & 0.227 & 0.8563 & 164 & 177 \\

& \multirow{3}{*}{4}
& Exact loader
& -- & -- & -- & -- & -- & 131039 & 65519 \\

&
& Direct TT
& 0.256 & 1.154 & 2.495 & 4.508 & 0.9826 & 519 & 274 \\

&
& Basket-CDF
& 0.052 & 0.044 & 0.089 & 0.156 & 0.6197 & 212 & 236 \\

\bottomrule
\end{tabular*}
\end{table}

\begin{table}[t]
\centering
\scriptsize
\setlength{\tabcolsep}{2pt}
\caption{Sample-target Basket-CDF loader results for $q=3$ uncertainty qubits per asset and dimensions $d=2,\ldots,5$. Market parameters are $S_0=50$, $\sigma=0.3$, $r=0.05$, and $T=1.0$. Results are shown for the equicorrelation case $\rho=0.6$ and the market-sector-style correlation case. Entries are multi-seed means over seeds $0,\ldots,4$ obtained with the loss weight $\lambda_B = 1000, \lambda_M = 10$, one latent qubit, $L_{\mathrm{dep}}=4$, sample size $N=10000$, and L-BFGS-B with 180 maximum iterations. Errors are relative pricing errors in percent with respect to the exact discretized reference.}
\label{tab:sampletarget_basketcdf_q3_mean}

\begin{tabular*}{\linewidth}{
    @{\extracolsep{\fill}}
    llccccccc
    @{}
}
\toprule
Case & $d$
& $K=40$ & $K=45$ & $K=50$ & $K=55$
& Fidelity & Depth & CX \\
\midrule

\multirow{4}{*}{Equicorr ($\rho=0.6$)}
& 2 & 0.837 & 1.067 & 1.413 & 1.812 & 0.8245 & 92 & 84 \\
& 3 & 0.766 & 1.062 & 1.501 & 1.968 & 0.7330 & 128 & 126 \\
& 4 & 0.368 & 0.459 & 0.723 & 0.630 & 0.5701 & 164 & 168 \\
& 5 & 1.458 & 1.882 & 2.498 & 2.779 & 0.4145 & 200 & 210 \\

\midrule

\multirow{4}{*}{Market-sector}
& 2 & 0.579 & 0.893 & 1.261 & 1.799 & 0.9963 & 92 & 84 \\
& 3 & 0.959 & 1.238 & 1.767 & 2.576 & 0.8673 & 128 & 126 \\
& 4 & 0.504 & 0.662 & 0.988 & 1.394 & 0.6454 & 164 & 168 \\
& 5 & 0.992 & 1.275 & 1.537 & 2.025 & 0.4681 & 200 & 210 \\

\bottomrule
\end{tabular*}
\end{table}

\begin{table}[t]
\centering
\scriptsize
\setlength{\tabcolsep}{2pt}
\caption{Bond-dimension sweep for Direct TT/MPS loading with $q=3$ uncertainty qubits per asset and dimensions $d=2,\ldots,5$. Market parameters are $S_0=50$, $\sigma=0.3$, $r=0.05$, and $T=1.0$. Results are shown for the equicorrelation case $\rho=0.6$ and the market-sector-style correlation case. The TT bond-dimension cap is denoted by $\chi$. Errors are relative pricing errors in percent with respect to the exact discretized reference. Entries below $0.001\%$ are reported as $<0.001$.}
\label{tab:directtt_rank_sweep_q3}

\begin{tabular*}{\linewidth}{
    @{\extracolsep{\fill}}
    lllccccccc
    @{}
}
\toprule
Case & $d$ & $\chi$
& $K=40$ & $K=45$ & $K=50$ & $K=55$
& Fidelity & Depth & CX \\
\midrule

\multirow{8}{*}{Equicorr ($\rho=0.6$)}
& \multirow{2}{*}{2}
& 8
& $<0.001$ & $<0.001$ & $<0.001$ & $<0.001$
& 1.0000 & 285 & 148 \\

&
& 16
& $<0.001$ & $<0.001$ & $<0.001$ & $<0.001$
& 1.0000 & 285 & 148 \\

& \multirow{2}{*}{3}
& 8
& 0.005 & 0.005 & 0.026 & 0.036
& 0.9999 & 851 & 436 \\

&
& 16
& $<0.001$ & $<0.001$ & $<0.001$ & $<0.001$
& 1.0000 & 2147 & 1092 \\

& \multirow{2}{*}{4}
& 8
& 0.019 & 0.045 & 0.096 & 0.147
& 0.9994 & 1413 & 726 \\

&
& 16
& $<0.001$ & $<0.001$ & $<0.001$ & $<0.001$
& 1.0000 & 4657 & 2366 \\

& \multirow{2}{*}{5}
& 8
& 0.041 & 0.123 & 0.223 & 0.375
& 0.9983 & 1975 & 1013 \\

&
& 16
& $<0.001$ & $<0.001$ & $<0.001$ & $<0.001$
& 1.0000 & 7156 & 3637 \\

\midrule

\multirow{8}{*}{Market-sector}
& \multirow{2}{*}{2}
& 8
& $<0.001$ & $<0.001$ & $<0.001$ & $<0.001$
& 1.0000 & 285 & 148 \\

&
& 16
& $<0.001$ & $<0.001$ & $<0.001$ & $<0.001$
& 1.0000 & 285 & 148 \\

& \multirow{2}{*}{3}
& 8
& $<0.001$ & $<0.001$ & $<0.001$ & $<0.001$
& 1.0000 & 845 & 435 \\

&
& 16
& $<0.001$ & $<0.001$ & $<0.001$ & $<0.001$
& 1.0000 & 2148 & 1092 \\

& \multirow{2}{*}{4}
& 8
& 0.026 & 0.051 & 0.110 & 0.172
& 0.9996 & 1409 & 727 \\

&
& 16
& $<0.001$ & $<0.001$ & $<0.001$ & $<0.001$
& 1.0000 & 4638 & 2367 \\

& \multirow{2}{*}{5}
& 8
& 0.063 & 0.084 & 0.098 & 0.261
& 0.9974 & 1975 & 1015 \\

&
& 16
& 0.002 & 0.002 & 0.002 & 0.008
& 1.0000 & 7145 & 3637 \\

\bottomrule
\end{tabular*}
\end{table}

\begin{figure}[t]
\centering
\includegraphics[
    width=0.62\textwidth
]{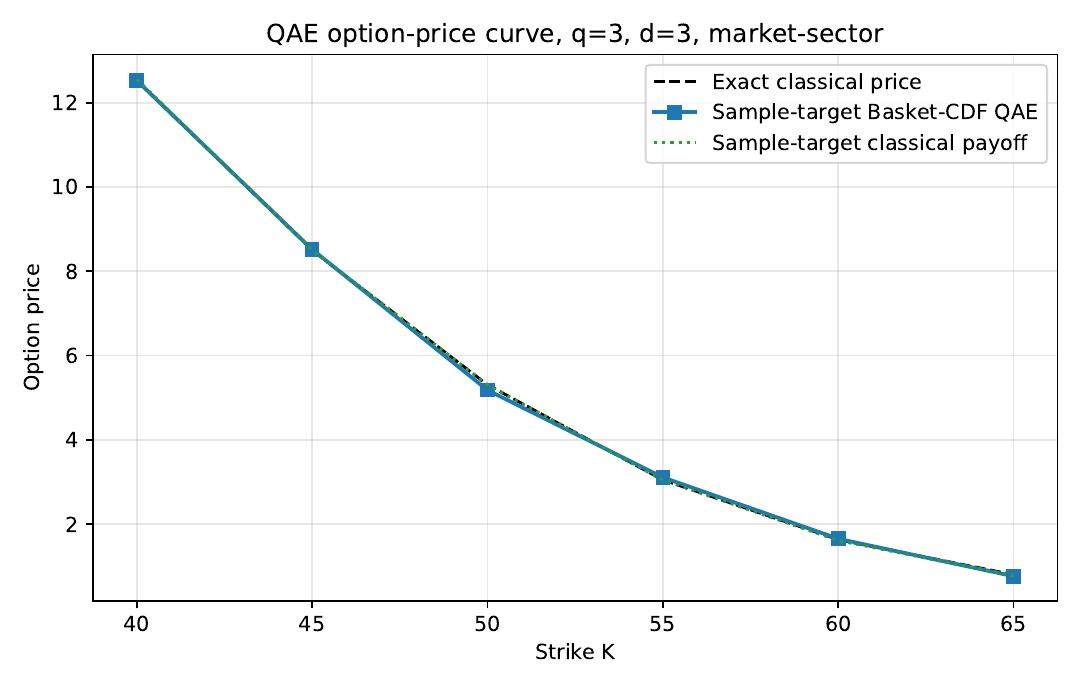}
\caption{QAE integration experiment using a Basket-CDF loader obtained from sampling-based training. The exact discretized price is computed from the target distribution, the classical payoff curve is computed from the distribution induced by the sampling-trained loader, and the QAE estimates are obtained by inserting the same learned loader into a standard QAE-based pricing workflow.}
\label{fig:qae_sanity_check}
\end{figure}

\FloatBarrier
    
    \section{Empirical-scenario Robustness}
    \label{app:empirical_scenario}
    
    This appendix reports an empirical-scenario evaluation of the Basket-CDF
    state-preparation principle. The goal is to test whether the loader
    can prepare a quantum state whose induced basket distribution matches an
    empirical scenario distribution, rather than a closed-form model-based
    probability mass function. This experiment is not intended as an arbitrage-free
    option-pricing calibration: historical return samples are observed under the
    physical measure, whereas market option prices are risk-neutral expectations.
    
    We use daily historical returns for four liquid U.S. equities, AAPL, MSFT, JPM,
    and XOM, over the period from May 17, 2021 to May 17, 2026.\footnote{
    Historical price data were downloaded using the \texttt{yfinance} Python
    package. 
    } Each asset is
    normalized to an initial value of 100, and each daily return observation is
    converted into a one-period terminal scenario. For asset \(i\), the empirical
    terminal values are discretized into \(2^q\) bins using marginal quantile
    binning, with \(q=3\) in this experiment. This produces an empirical four-asset
    grid distribution over \(2^{qd}=2^{12}\) basis states. The one-asset marginals
    are estimated directly from the binned samples, and the target Basket-CDF is
    estimated from the empirical basket values
    $B(x)=\frac{1}{d}\sum_{i=1}^d S_i(x_i).$

    Thus, the training targets are empirical marginal histograms and empirical
    basket-CDF values, rather than a closed-form joint probability mass function.
    
    The same two-stage architecture is used as in the model-based experiments.
    First, the asset-wise marginal distributions are loaded by TT-informed marginal
    circuits. Second, a one-qubit latent dependence block is trained using the
    Basket-CDF and marginal losses. We use the single-pass shallow dependence block,
    in which each layer applies one latent-to-asset controlled-rotation fanout
    followed by local asset mixing, without a second latent injection in the same
    layer. The resulting circuit is evaluated on the empirical basket payoff curve
    \[
    \mathbb{E}\bigl[(B-K)^+\bigr],
    \qquad
    K\in\{96,97,98,99,100\}.
    \]
    Again, this quantity should be interpreted as a historical scenario payoff
    functional, not as a risk-neutral option price.
    
    \begin{figure}[t]
    \centering
    \begin{subfigure}{0.48\textwidth}
        \centering
        \includegraphics[width=\textwidth]{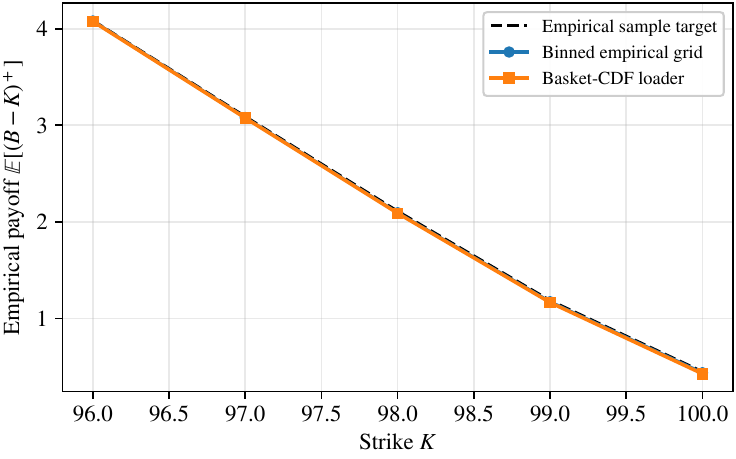}
        \caption{Empirical payoff curve.}
        \label{fig/fig:data_driven_payoff_curve}
    \end{subfigure}
    \hfill
    \begin{subfigure}{0.48\textwidth}
        \centering
        \includegraphics[width=\textwidth]{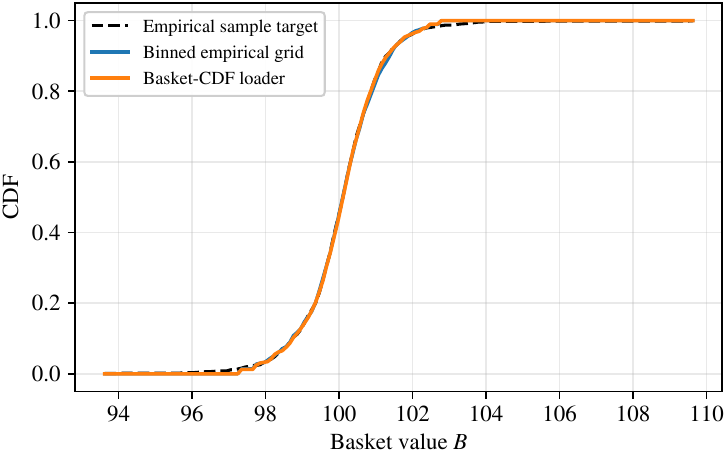}
        \caption{Empirical basket CDF.}
        \label{fig:data_driven_basket_cdf_curve}
    \end{subfigure}
    \caption{
    Empirical-scenario Basket-CDF state preparation using historical equity return
    scenarios. Daily returns of AAPL, MSFT, JPM, and XOM from May 17, 2021 to
    May 17, 2026 are converted into empirical terminal scenarios and binned onto a
    $q=3$, $d=4$ asset grid. The loader is trained from empirical marginal
    histograms and empirical Basket-CDF values. The experiment is intended as a
    scenario-distribution loading test, not as a risk-neutral option-pricing
    calibration.
    }
    \label{fig:data_driven_basket_cdf}
    \end{figure}
    
    Figure~\ref{fig:data_driven_basket_cdf} shows that the trained state reproduces
    both the empirical payoff curve and the empirical basket CDF closely. Over the
    five strikes \(K=96,\ldots,100\), the single-pass latent block achieves a mean
    relative payoff-curve error of \(1.56\%\) and a maximum relative error of
    \(4.57\%\) against the empirical sample payoff curve. The corresponding
    state-preparation circuit has depth 164 and 168 CX gates. These results indicate
    that the proposed Basket-CDF objective can be constructed from empirical
    samples while retaining accurate basket-level behavior in the prepared state. Thus, this experiment should be viewed as a scenario-distribution loading test.
It shows that the Basket-CDF mechanism can be trained from sample-based
financial targets, with possible use in payoff curves, tail probabilities, VaR,
or CVaR-type risk measures.
\end{appendices}


\bibliography{sn-bibliography}

@article{stamatopoulos2020option,
  title={Option pricing using quantum computers},
  author={Stamatopoulos, Nikitas and Egger, Daniel J and Sun, Yue and Zoufal, Christa and Iten, Raban and Shen, Ning and Woerner, Stefan},
  journal={Quantum},
  volume={4},
  pages={291},
  year={2020},
  publisher={Verein zur F{\"o}rderung des Open Access Publizierens in den Quantenwissenschaften},
  note = {\url{https://doi.org/10.22331/q-2020-07-06-291}}
}

@article{zoufal2019quantum,
  title={Quantum generative adversarial networks for learning and loading random distributions},
  author={Zoufal, Christa and Lucchi, Aur{\'e}lien and Woerner, Stefan},
  journal={npj Quantum Information},
  volume={5},
  number={1},
  pages={103},
  year={2019},
  publisher={Nature Publishing Group UK London},
  note = {\url{https://doi.org/10.1038/s41534-019-0223-2}}
}

@article{herman2026quantum,
  title={Quantum Speedups for Derivative Pricing Beyond {Black-Scholes}},
  author={Herman, Dylan and Sun, Yue and Liu, Jin-Peng and Pistoia, Marco and Che, Charlie and Otter, Rob and Chakrabarti, Shouvanik and Harrow, Aram},
  journal={arXiv preprint arXiv:2602.03725},
  year={2026},
  note = {\url{ 	
https://doi.org/10.48550/arXiv.2602.03725}}
}

@article{alcazar2022quantum,
  title={Quantum algorithm for credit valuation adjustments},
  author={Alcazar, Javier and Cadarso, Andrea and Katabarwa, Amara and Mauri, Marta and Peropadre, Borja and Wang, Guoming and Cao, Yudong},
  journal={New Journal of Physics},
  volume={24},
  number={2},
  pages={023036},
  year={2022},
  publisher={IOP Publishing},
  note = {\url{https://doi.org/10.1088/1367-2630/ac5003}}
}

@article{iaconis2024quantum,
  title={Quantum state preparation of normal distributions using matrix product states},
  author={Iaconis, Jason and Johri, Sonika and Zhu, Elton Yechao},
  journal={npj Quantum Information},
  volume={10},
  number={1},
  pages={15},
  year={2024},
  publisher={Nature Publishing Group UK London},
  note = {\url{https://doi.org/10.1038/s41534-024-00805-0}}
}

@article{yu2026quantum,
  title={Quantum advantage for multi-option portfolio pricing and valuation adjustments},
  author={Han, Jeong Yu and Cheng, Bin and Vu, Dinh-Long
             and Rebentrost, Patrick},
  journal={Quantitative Finance},
  volumn = {26},
  number = {3},
  pages = {467--489},
  year={2026},
  publisher={Taylor \& Francis},
  note = {\url{https://doi.org/10.1080/14697688.2026.2614573}}
}

@article{melnikov2023quantum,
  title={Quantum state preparation using tensor networks},
  author={Melnikov, Ar A and Termanova, Alena A and Dolgov, Sergey V and Neukart, Florian and Perelshtein, MR},
  journal={Quantum Science and Technology},
  volume={8},
  number={3},
  pages={035027},
  year={2023},
  publisher={IOP Publishing},
  note = {\url{https://doi.org/10.1088/2058-9565/acd9e7}}
}

@article{cerezo2021variational,
  title={Variational quantum algorithms},
  author={Cerezo, Marco and Arrasmith, Andrew and Babbush, Ryan and Benjamin, Simon C and Endo, Suguru and Fujii, Keisuke and McClean, Jarrod R and Mitarai, Kosuke and Yuan, Xiao and Cincio, Lukasz and others},
  journal={Nature Reviews Physics},
  volume={3},
  number={9},
  pages={625--644},
  year={2021},
  publisher={Nature Publishing Group UK London},
  note = {\url{https://doi.org/10.1038/s42254-021-00348-9}}
}

@article{ran2020encoding,
  title={Encoding of matrix product states into quantum circuits of one-and two-qubit gates},
  author={Ran, Shi-Ju},
  journal={Physical Review A},
  volume={101},
  number={3},
  pages={032310},
  year={2020},
  publisher={APS},
  note = {\url{https://doi.org/10.1103/PhysRevA.101.032310}}
}

@article{grinko2021iterative,
  title={Iterative quantum amplitude estimation},
  author={Grinko, Dmitry and Gacon, Julien and Zoufal, Christa and Woerner, Stefan},
  journal={npj Quantum Information},
  volume={7},
  number={1},
  pages={52},
  year={2021},
  publisher={Nature Publishing Group UK London},
  note = {\url{https://doi.org/10.1038/s41534-021-00379-1}}
}

@article{pereira2024encoding,
  title={Encoding of Probability Distributions for Quantum {Monte Carlo} Using Tensor Networks},
  author={Pereira, Antonio and Villarino, Alba and Cortines, Aser and Mugel, Samuel and Orus, Roman and Beltran, Victor Leme and Scursulim, JVS and Brito, Samurai},
  journal={arXiv preprint arXiv:2411.11660},
  year={2024},
  note = {\url{ 	
https://doi.org/10.48550/arXiv.2411.11660}}
}

@article{aaronson2015read,
  title={Read the fine print},
  author={Aaronson, Scott},
  journal={Nature Physics},
  volume={11},
  number={4},
  pages={291--293},
  year={2015},
  publisher={Nature Publishing Group UK London},
  note = {\url{https://doi.org/10.1038/nphys3272}}
}

@article{brassard2002quantum,
  title={Quantum amplitude amplification and estimation},
  author={Brassard, Gilles and Hoyer, Peter and Mosca, Michele and Tapp, Alain},
  journal={arXiv preprint quant-ph/0005055},
  year={2000},
  note = {\url{https://doi.org/10.48550/arXiv.quant-ph/0005055
}}
}

@article{suzuki2020amplitude,
  title={Amplitude estimation without phase estimation},
  author={Suzuki, Yohichi and Uno, Shumpei and Raymond, Rudy and Tanaka, Tomoki and Onodera, Tamiya and Yamamoto, Naoki},
  journal={Quantum Information Processing},
  volume={19},
  number={2},
  pages={75},
  year={2020},
  publisher={Springer},
  note = {\url{https://doi.org/10.1007/s11128-019-2565-2}}
}

@article{preskill2018quantum,
  title={Quantum computing in the {NISQ} era and beyond},
  author={Preskill, John},
  journal={Quantum},
  volume={2},
  pages={79},
  year={2018},
  publisher={Verein zur F{\"o}rderung des Open Access Publizierens in den Quantenwissenschaften},
  note = {\url{https://doi.org/10.22331/q-2018-08-06-79}}
}

@inproceedings{murali2019noise,
  title={Noise-adaptive compiler mappings for noisy intermediate-scale quantum computers},
  author={Murali, Prakash and Baker, Jonathan M and Javadi-Abhari, Ali and Chong, Frederic T and Martonosi, Margaret},
  booktitle={Proceedings of the twenty-fourth international conference on architectural support for programming languages and operating systems},
  pages={1015--1029},
  year={2019},
  note = {\url{https://doi.org/10.1145/3297858.3304075}}
}

@article{kandala2017hardware,
  title={Hardware-efficient variational quantum eigensolver for small molecules and quantum magnets},
  author={Kandala, Abhinav and Mezzacapo, Antonio and Temme, Kristan and Takita, Maika and Brink, Markus and Chow, Jerry M and Gambetta, Jay M},
  journal={Nature},
  volume={549},
  number={7671},
  pages={242--246},
  year={2017},
  publisher={Nature Publishing Group},
  note = {\url{https://doi.org/10.1038/nature23879}}
}

@article{tkachenko2021correlation,
  title={Correlation-informed permutation of qubits for reducing ansatz depth in the variational quantum eigensolver},
  author={Tkachenko, Nikolay V and Sud, James and Zhang, Yu and Tretiak, Sergei and Anisimov, Petr M and Arrasmith, Andrew T and Coles, Patrick J and Cincio, Lukasz and Dub, Pavel A},
  journal={PRX Quantum},
  volume={2},
  number={2},
  pages={020337},
  year={2021},
  publisher={APS},
  note = {\url{https://doi.org/10.1103/PRXQuantum.2.020337}}
}

@article{woerner2019quantum,
  title={Quantum risk analysis},
  author={Woerner, Stefan and Egger, Daniel J},
  journal={npj Quantum Information},
  volume={5},
  number={1},
  pages={15},
  year={2019},
  publisher={Nature Publishing Group UK London},
  note = {\url{https://doi.org/10.1038/s41534-019-0130-6}}
}

@article{manzano2025alternative,
  title={Alternative pipeline for option pricing using quantum computers},
  author={Manzano, Alberto and Ferro, Gonzalo and Leitao, {\'A}lvaro and V{\'a}zquez, Carlos and G{\'o}mez, Andr{\'e}s},
  journal={EPJ Quantum Technology},
  volume={12},
  pages={28},
  year={2025},
  publisher={Springer},
  note = {\url{https://doi.org/10.1140/epjqt/s40507-025-00328-3}}
}

@article{oseledets2011tensor,
  title={Tensor-train decomposition},
  author={Oseledets, Ivan V},
  journal={SIAM Journal on Scientific Computing},
  volume={33},
  number={5},
  pages={2295--2317},
  year={2011},
  publisher={SIAM},
  note = {\url{https://doi.org/10.1137/090752286}}
}

@article{kaneko2022quantum,
  title={Quantum pricing with a smile: implementation of local volatility model on quantum computer},
  author={Kaneko, Kazuya and Miyamoto, Koichi and Takeda, Naoyuki and Yoshino, Kazuyoshi},
  journal={EPJ Quantum Technology},
  volume={9},
  pages={7},
  year={2022},
  publisher={Springer},
  note = {\url{https://doi.org/10.1140/epjqt/s40507-022-00125-2}}
}

@article{miyamoto2022bermudan,
  title={Bermudan option pricing by quantum amplitude estimation and {Chebyshev} interpolation},
  author={Miyamoto, Koichi},
  journal={EPJ Quantum Technology},
  volume={9},
  pages={3},
  year={2022},
  publisher={Springer},
  note = {\url{https://doi.org/10.1140/epjqt/s40507-022-00124-3}}
}

@article{manzano2023real,
  title={Real quantum amplitude estimation},
  author={Manzano, Alberto and Musso, Daniele and Leitao, {\'A}lvaro},
  journal={EPJ Quantum Technology},
  volume={10},
  pages={2},
  year={2023},
  publisher={Springer},
  note = {\url{https://doi.org/10.1140/epjqt/s40507-023-00159-0}}
}

@article{miyamoto2024bias,
  title={On the bias in iterative quantum amplitude estimation},
  author={Miyamoto, Koichi},
  journal={EPJ Quantum Technology},
  volume={11},
  pages={42},
  year={2024},
  publisher={Springer},
  note = {\url{https://doi.org/10.1140/epjqt/s40507-024-00253-x}}
}

@article{benedetti2019generative,
  title={A generative modeling approach for benchmarking and training shallow quantum circuits},
  author={Benedetti, Marcello and Garcia-Pintos, Delfina and Perdomo, Oscar and Leyton-Ortega, Vicente and Nam, Yunseong and Perdomo-Ortiz, Alejandro},
  journal={npj Quantum Information},
  volume={5},
  number={1},
  pages={45},
  year={2019},
  publisher={Nature Publishing Group UK London},
  note = {\url{https://doi.org/10.1038/s41534-019-0157-8}}
}

@article{matheson1976scoring,
  title={Scoring rules for continuous probability distributions},
  author={Matheson, James E and Winkler, Robert L},
  journal={Management Science},
  volume={22},
  number={10},
  pages={1087--1096},
  year={1976},
  publisher={INFORMS},
  note = {\url{https://doi.org/10.1287/mnsc.22.10.1087}}
}

@article{gneiting2007strictly,
  title={Strictly proper scoring rules, prediction, and estimation},
  author={Gneiting, Tilmann and Raftery, Adrian E},
  journal={Journal of the American Statistical Association},
  volume={102},
  number={477},
  pages={359--378},
  year={2007},
  publisher={Taylor \& Francis},
  note = {\url{https://doi.org/10.1198/016214506000001437}}
}

@article{caldana2016general,
  title={General closed-form basket option pricing bounds},
  author={Caldana, Ruggero and Fusai, Gianluca and Gnoatto, Alessandro and Grasselli, Martino},
  journal={Quantitative Finance},
  volume={16},
  number={4},
  pages={535--554},
  year={2016},
  publisher={Taylor \& Francis},
  note = {\url{https://doi.org/10.1080/14697688.2015.1073854}}
}

@article{rosenberg1973prediction,
  title={The prediction of systematic and specific risk in common stocks},
  author={Rosenberg, Barr and McKibben, Walt},
  journal={Journal of Financial and Quantitative Analysis},
  volume={8},
  number={2},
  pages={317--333},
  year={1973},
  publisher={Cambridge University Press},
  note = {\url{https://doi.org/10.2307/2330027 }}
}

@article{fama1993common,
  title={Common risk factors in the returns on stocks and bonds},
  author={Fama, Eugene F and French, Kenneth R},
  journal={Journal of Financial Economics},
  volume={33},
  number={1},
  pages={3--56},
  year={1993},
  publisher={Elsevier},
  note = {\url{https://doi.org/10.1016/0304-405X(93)90023-5}}
}

@article{carhart1997persistence,
  title={On persistence in mutual fund performance},
  author={Carhart, Mark M},
  journal={The Journal of Finance},
  volume={52},
  number={1},
  pages={57--82},
  year={1997},
  publisher={Wiley Online Library},
  note = {\url{https://doi.org/10.1111/j.1540-6261.1997.tb03808.x}}
}

@article{carrera2021efficient,
  title={Efficient state preparation for quantum amplitude estimation},
  author={Carrera Vazquez, Almudena and Woerner, Stefan},
  journal={Physical Review Applied},
  volume={15},
  number={3},
  pages={034027},
  year={2021},
  publisher={APS},
  note = {\url{https://doi.org/10.1103/PhysRevApplied.15.034027}}
}

@article{gomez2022survey,
  title={A {Survey} on {Quantum Computational Finance} for {Derivatives} {Pricing} and {VaR}},
  author={G{\'o}mez, Andr{\'e}s and Leitao, {\'A}lvaro and Manzano, Alberto and Musso, Daniele and Nogueiras, Mar{\'\i}a R and Ord{\'o}{\~n}ez, Gustavo and V{\'a}zquez, Carlos},
  journal={Archives of {Computational} {Methods} in {Engineering}},
  volume={29},
  number={6},
  pages={4137--4163},
  year={2022},
  publisher={Springer},
  note = {\url{https://doi.org/10.1007/s11831-022-09732-9}}
}

@article{rattew2021efficient,
  title={The efficient preparation of normal distributions in quantum registers},
  author={Rattew, Arthur G and Sun, Yue and Minssen, Pierre and Pistoia, Marco},
  journal={Quantum},
  volume={5},
  pages={609},
  year={2021},
  publisher={Verein zur F{\"o}rderung des Open Access Publizierens in den Quantenwissenschaften},
  note = {\url{https://doi.org/10.22331/q-2021-12-23-609}}
}

@inproceedings{mori2024quantum,
  title={Quantum algorithm for copula-based risk aggregation using orthogonal series density estimation},
  author={Mori, Hitomi and Miyamoto, Koichi},
  booktitle={2024 IEEE International Conference on Quantum Computing and Engineering (QCE)},
  volume={1},
  pages={322--330},
  year={2024},
  organization={IEEE},
  note = {\url{https://doi.org/10.1109/QCE60285.2024.00046}}
}

@article{zhu2023copula,
  title={Copula-based risk aggregation with trapped ion quantum computers},
  author={Zhu, Daiwei and Shen, Weiwei and Giani, Annarita and Ray-Majumder, Saikat and Neculaes, Bogdan and Johri, Sonika},
  journal={Scientific Reports},
  volume={13},
  number={1},
  pages={18511},
  year={2023},
  publisher={Nature Publishing Group UK London},
  note = {\url{https://doi.org/10.1038/s41598-023-44151-1}}
}

@article{zhu2022generative,
  title={Generative quantum learning of joint probability distribution functions},
  author={Zhu, Elton Yechao and Johri, Sonika and Bacon, Dave and Esencan, Mert and Kim, Jungsang and Muir, Mark and Murgai, Nikhil and Nguyen, Jason and Pisenti, Neal and Schouela, Adam and others},
  journal={Physical Review Research},
  volume={4},
  number={4},
  pages={043092},
  year={2022},
  publisher={APS},
  note = {\url{https://doi.org/10.1103/PhysRevResearch.4.043092}}
}

@article{dvoretzky1956asymptotic,
  title={Asymptotic minimax character of the sample distribution function and of the classical multinomial estimator},
  author={Dvoretzky, Aryeh and Kiefer, Jack and Wolfowitz, Jacob},
  journal={The Annals of Mathematical Statistics},
  volume  = {27},
  number  = {3},
  pages={642--669},
  year={1956},
  publisher={JSTOR}
}

@article{massart1990tight,
  title={The tight constant in the {Dvoretzky-Kiefer-Wolfowitz} inequality},
  author={Massart, Pascal},
  journal={The Annals of Probability},
  volume  = {18},
  number  = {3},
  pages={1269--1283},
  year={1990},
  publisher={JSTOR},
  note = {\url{https://www.jstor.org/stable/2244426}}
}

@article{kolla2019concentration,
  title={Concentration bounds for empirical conditional value-at-risk: The unbounded case},
  author={Kolla, Ravi Kumar and Prashanth, LA and Bhat, Sanjay P and Jagannathan, Krishna},
  journal={Operations Research Letters},
  volume={47},
  number={1},
  pages={16--20},
  year={2019},
  publisher={Elsevier},
  note = {\url{https://doi.org/10.1016/j.orl.2018.11.005}}
}

@article{rudolph2024trainability,
  title={Trainability barriers and opportunities in quantum generative modeling},
  author={Rudolph, Manuel S and Lerch, Sacha and Thanasilp, Supanut and Kiss, Oriel and Shaya, Oxana and Vallecorsa, Sofia and Grossi, Michele and Holmes, Zo{\"e}},
  journal={npj Quantum Information},
  volume={10},
  number={1},
  pages={116},
  year={2024},
  publisher={Nature Publishing Group UK London},
  note = {\url{https://doi.org/10.1038/s41534-024-00902-0}}
}

@article{tian2025quantum,
  title={Quantum generative adversarial network with automated noise suppression mechanism based on {WGAN-GP}},
  author={Tian, Yanbing and Tian, Cewen and Fan, Zaixu and Fu, Minghao and Ma, Hongyang},
  journal={EPJ Quantum Technology},
  volume={12},
  pages={80},
  year={2025},
  publisher={Springer},
  note = {\url{https://doi.org/10.1140/epjqt/s40507-025-00372-z}}
}

@article{blank2021quantum,
  title={Quantum-enhanced analysis of discrete stochastic processes},
  author={Blank, Carsten and Park, Daniel K. and Petruccione, Francesco},
  journal={npj Quantum Information},
  volume={7},
  number={1},
  pages={126},
  year={2021},
  publisher={Nature Publishing Group},
  note = {\url{https://doi.org/10.1038/s41534-021-00459-2}}
}

\end{document}